\documentclass[journal]{IEEEtran}

\IEEEoverridecommandlockouts                              % This command is only needed if 
\usepackage{mathptmx} % assumes new font selection scheme installed
\usepackage{amsmath} % assumes amsmath package installed
\usepackage{amssymb}  % assumes amsmath package installed
\usepackage[dvipsnames]{xcolor}

\usepackage{tikz} % drawing vector graphics
\usetikzlibrary{positioning, decorations.markings, arrows, arrows.meta}
\usepackage{forest} % drawing trees
\usepackage{wrapfig} % wrapping text around figures
\usepackage{hyperref} % hyperlinking urls and references
\usepackage[flushmargin,hang]{footmisc} % compress footnotes
\usepackage{cite} % compress numeric citations
\usepackage{subcaption} % subfigures

\usepackage{amsthm} % theorem styles
\usepackage{cancel}
\usepackage{dirtytalk}
\usepackage{amsmath}
\usepackage{flushend}

\newtheorem{theorem}{Theorem}
\newtheorem{definition}{Definition}
\newtheorem{lemma}{Lemma}
\newtheorem{example}{Example}
\newtheorem{corollary}{Corollary}
\newtheorem{remark}{Remark}
\newtheorem{assumption}{Assumption}

\newcommand{\bt}{\mathcal{T}}

\title{\LARGE \bf
An Extended Convergence Result for Behaviour Tree Controllers
}

\author{Christopher Iliffe Sprague \qquad\qquad Petter \"Ogren% <-this % stops a space
\thanks{Christopher Iliffe Sprague and Petter \"Ogren are with the Robotics, Perception and Learning Lab., School of Electrical Engineering and Computer Science, 
        Royal Institute of Technology (KTH), SE-100 44 Stockholm, Sweden, {\tt\small sprague@kth.se}}%
}

\begin{document}

\maketitle
\thispagestyle{empty}
\pagestyle{empty}

%%%%%%%%%%%%%%%%%%%%%%%%%%%%%%%%%%%%%%%%%%%%%%%%%%%%%%%%%%%%%%%%%%%%%%%%%%%%%%%%
\begin{abstract}
    Behavior trees (BTs) are an optimally modular framework to assemble hierarchical hybrid control policies from a set of low-level control policies using a tree structure. Many robotic tasks are naturally decomposed into a hierarchy of control tasks, and modularity is a well-known tool for handling complexity, therefor behavior trees have garnered widespread usage in the robotics community. In this paper, we study the convergence of BTs, in the sense of reaching a desired part of the state space. Earlier results on BT convergence were often tailored to specific families of BTs, created using different design principles. The results of this paper generalize the earlier results and also include new cases of cyclic switching not covered in the literature.
\end{abstract}

    
% \begin{keywords}
%     %List of keywords (from the RA Letters keyword list)
%     Behaviour-Based Systems,  Robot Safety, Control Architectures and Programming
%     % Hybrid Logical/Dynamical Planning and Verification, Formal Methods in Robotics and Automation,
% \end{keywords}
    
    %\begin{keywords}
    %Neural Networks, Machine Learning, Optimal Control, Online Planning
    %\end{keywords}
    
    %paxton_evaluating_2018\cite{iovino2020survey}
\section{Introduction}

% BTs are...
% \IEEEPARstart{B}{ehavior trees} (BTs) 
Behavior trees (BTs)
are a way to assemble a hierarchical hybrid control policy (HCP) from a set of low-level control policies using a tree structure.
They can be analysed at any level of hierarchy --- low-level control policies at the lowest level and modules of control policies (forming HCPs) at subsequently higher levels of hierarchy.
At every level of hierarchy, these modules interface with each other in an identical functional way.
These aspects allow individual control policies and modules to be developed and tested independently of others.

The modular development that BTs allow is, in fact, the reason they were conceived in the first place in the video game industry \cite{islahandling2005}.
In the virtual context of video games, where the world is predictable by design, the design of low-level control policies is often trivial, and
therefore, game designers found themselves looking for a modular way to compose a large set of low-level control policies into rich behaviors. %hence game designers felt a pronounced need to more easily compose large sets of different low-level policies into rich behaviors.

In contrast, the robotics community faces far more challenges in the design of low-level control policies, as the real world can only be approximated by models, and tasks such as grasping are important research fields of their own. Thus, while focusing on creating control policies for solving fundamental interaction problems, robotics researchers felt less of a need for a hierarchical modular framework for combining many of such control policies.
However, with the surge in collaborative development brought forth by open-source platforms such as the Robotic Operating System (ROS), the increasingly ubiquitous presence of robots in society, and the rise of control-policy synthesis through machine learning \cite{brunke2021safe}, large sets of low-level control policies are becoming increasingly available for many robotic systems.
As a result, in recent years, the robotics community has begun to feel the same need for a modular framework as the video game industry.

A popular choice for representing HCPs has been hybrid automata (HA).
However, the discrete states of HAs are modeled by finite-state machines (FSMs), which rely on state transitions.
As such, the development of individual control policies or modules in an FSM of $N$ states requires taking into account the $N(N-1)$ possible state transitions between the $N$ states.
In contrast, modules in BTs rely on a common functional interface, with return values representing current progress as either Success, Failure or Running (see below), at every level of hierarchy, allowing each module at any level of hierarchy in any subtree to be developed individually --- the explicit transitions of a FSM are thus implicitly encoded in the functional interface of BTs \cite{colledanchise2016behavior}.
For this reason, BTs have been shown to be optimally modular \cite{biggar2020modularity}, while at the same time equally as expressive as FSMs \cite{biggar2021expressiveness}.

\begin{figure}[t]
    \centering
    \includegraphics[width=\linewidth]{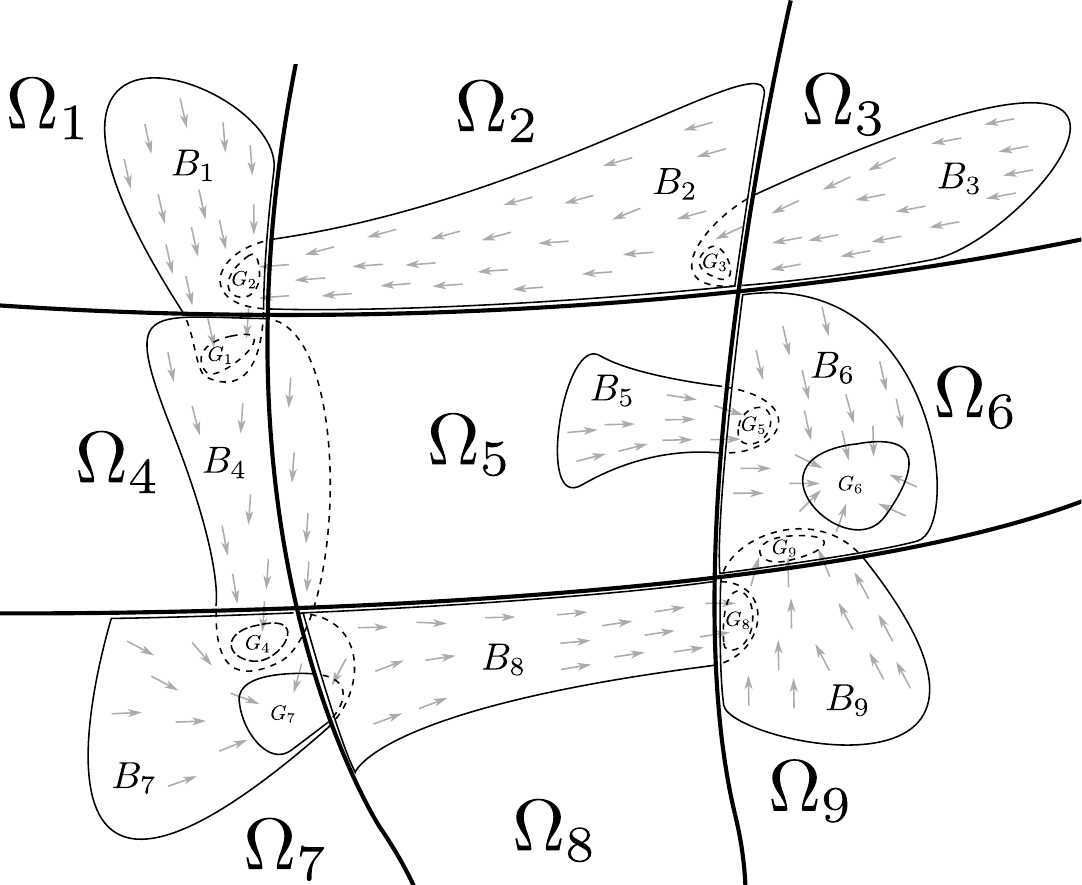}
    \caption{An example of how the operating regions $\Omega_i$ and domains of attraction $B_i$ can be arranged to make the state reach one of the desired goal regions $G_6,G_7$ from a wide variety of starting states.}
    \label{fig:operating_regions}
\end{figure}

    In the past years, BTs have garnered an increasing use in the robotics community, with over 100 papers in recent surveys \cite{iovino2020survey} and software by innovative companies such as
    Boston Dynamics\footnote{\url{https://dev.bostondynamics.com/docs/concepts/autonomy/missions_service}},     Nvidia\footnote{\url{https://docs.nvidia.com/isaac/packages/behavior_tree/doc/behavior_trees.html}}, and Google's Intrinsic.AI\footnote{\url{https://intrinsic.ai/blog/posts/introducing-intrinsic-flowstate/}}.
    
    Several BT design principles have been proposed, e.g. implicit sequences \cite{colledanchise2016behavior} and backchaining \cite{colledanchise2019towards, ogren2020convergence}, to create HCPs that are directed at achieving some top-level goals.
    These designs qualitatively resemble the notion of \say{sequential composition} from \cite{burridge1999sequential} --- the composition of low-level control policies in such a way as to enlarge the system's domain of attraction (DOA) to a goal state, as described in \cite{colledanchise2016behavior}.
    
    Relatively few works have \textit{formally} investigated the convergence properties of BT designs \cite{colledanchise2016behavior, %sprague2018adding,
    ogren2020convergence, sprague2021continuous},
    and out of these, none cover the case where there are cycles in the execution of control policies.
    %In \cite{sprague2021continuous}, sufficient conditions under which a general BT will be convergent to a desired region of the state space were presented when there are not any cycles in the execution of control policies.

    \paragraph*{\textbf{Contributions}}The three main contributions of this paper are as follows. 
    \begin{enumerate}
        \item We present the first general convergence theorem that allows cycles in the execution of sub-BTs and the execution of sub-BTs outside their DOA (Theorem \ref{theorem:general_convergence}).
        \item We show how this theorem also generalizes earlier proofs in \cite{colledanchise2016behavior,ogren2020convergence,sprague2021continuous}, in detail: 
        \cite[Lemma 2, p.7]{colledanchise2016behavior}  is covered by Corollary \ref{cor:sequence},
        \cite[Lemma 3, p.8]{colledanchise2016behavior} is covered by Corollary \ref{cor:implicit_sequence},
        \cite[Theorem 4, p.6]{sprague2021continuous} is covered by Corollary \ref{cor:sprague2021continuous},
        \cite[Theorem 1, p.7]{ogren2020convergence} is covered by Corollary \ref{cor:ogren2020convergence}.
        \item Finally, the result also encompasses the recent conference paper \cite{sprague2022adding} on combining black box controllers, created by i.e., reinforcement learning, with model-based designs without losing formal guarantees on safety and goal convergence.
    \end{enumerate}
    
    The organization of this letter is as follows. First, we describe related work in Section~\ref{sec:related-work}, then Section~\ref{sec:preliminaries} provides a fairly detailed background on a formal description of BTs that will be used for the convergence analysis. The main result can be found in Section~\ref{sec:main}, followed by the conclusions in Section~\ref{sec:conclusions}.

    \section{Related Work}\label{sec:related-work}
In this section, we will describe related work in three different categories.
    \paragraph*{\textbf{Hybrid dynamical system representations}}
    A hybrid dynamical system (HDS) is a dynamical system that has both continuous and discrete dynamics.
    Among the most popular ways to represent HDSs is hybrid automata (HA),
    where continuous behavior is given by ordinary differential equations and discrete behavior is given by finite-state machines (FSMs).
    However, in recent years, BTs have become a popular way to represent HDSs due to their modularity \cite{ogren2012increasing, marzinotto2014towards,colledanchise2016behavior}.
    This modularity allows different dynamics to be inserted into the HDS without having to take into account the myriad of transitions entailed by FSMs.
    BTs have also been theoretically proven to be optimally modular reactive control architectures \cite{biggar2020modularity}.
    \emph{The modularity of BTs is the main motivation behind the contributions of this letter.}

    BTs have been informally compared to HDSs several times \cite{ogren2012increasing, marzinotto2014towards,colledanchise2016behavior}.
    The first formal comparison of BTs to HDSs was given in \cite{sprague2021continuous}, where BTs where formulated as discontinuous dynamical systems (DDSs), a subclass of HDSs.
    \emph{In this paper, we use the DDS formalism of \cite{sprague2021continuous} to analyse convergence.}

    \paragraph*{\textbf{Sequential composition}}
    As mentioned before, BTs are a way to compose together multiple different control policies.
    This concept is well-represented in the literature \cite{burridge1999sequential, conner2006integrated,tedrake2010lqr,Konidaris2012RobotLF,colledanchise2016behavior}, and is rooted in so-called sequential composition \cite{burridge1999sequential}.
    The main idea is that the execution of a control policy should lead the system's state into the domain of attraction (DOA) of another control policy.
    If a set of such pairs of control policies can be \say{backchained} from a goal region, then the overall system can be stabilized to the goal starting within the DOA of any control policy in the set.
    The collection of all such chains induces a so-called \say{prepares graph}, which represents all the possible transitions between control policies in discrete space.
    As noted in \cite{conner2006integrated}, this representation functions as a finite-state machine (FSM), which, when factoring in the continuous behavior between discrete states, functions as a hybrid automaton \cite{lygeros2003dynamical}.
    But, as noted above, BTs have advantages over FSMs in terms of modularity \cite{biggar2020modularity}.
    Luckily, BTs were shown to generalize sequential composition in \cite{colledanchise2016behavior}.
    \emph{But, unlike \cite{colledanchise2016behavior}, in this paper, we show how sequential composition can be induced by BTs with the DDS formalism of \cite{sprague2021continuous}, which allows the application of existence and uniqueness results of DDSs \cite{cortes2008discontinuous}.}

    A commonality among many works that deal with sequential composition is that the control policies used to create the composition all have some indication of which other control policy they can \say{switch} to, based on their DOA.
    A notion of composing together control policies that do not necessarily have such an indication is given in \cite{paxton2019representing}, in the context of specific BT designs, where probabilities of \say{uncontrolled switches} were analyzed.
    \emph{Unlike \cite{paxton2019representing}, we analyze uncontrolled switches in general BT designs.}

    \paragraph*{\textbf{Stability}}
    In many cases, BTs are used to model hybrid control policies (HCPs), which when imparted into a given continuous dynamical system, renders a HDS.
    Formal guarantees of stability is a key subject in HDSs \cite{branicky1998multiple,ye1998stability,liberzon2003switching,goebel2012hybrid} because it states what the system will do (i.e. through stability to a goal region) and what it will not do (i.e. through set-invariance within an obstacle-free region).
    Typically, when executing a BT as a HCP, one seeks to stabilize the system to the so-called success region, i.e. some desired region of the statespace.
    The stability (or convergence) of BTs has been studied in several works \cite{colledanchise2016behavior, ogren2020convergence,paxton2019representing,rovida2017extended}.
    However, these works only consider specific BT designs, not general ones.
    In \cite{sprague2021continuous}, stability was analyzed for general BT designs using the DDS formalism, but did not allow for cycles.
    \emph{In this paper, we show how the presented theorem (Theorem \ref{theorem:general_convergence}) generalizes the stability theorems of
    \cite[Lemma 2, p.7]{colledanchise2016behavior},
    \cite[Lemma 3, p.8]{colledanchise2016behavior},
    \cite[Theorem 4, p.6]{sprague2021continuous}, and 
    \cite[Theorem 1, p.7]{ogren2020convergence}.}

    % ; however, only for the case where the control policy executions are partially ordered and the corresponding prepares graph is acyclical.
    % In BTs where there is supposed to be some reaccuring behavior, such as surveying an area with a robot until it needs to recharge \cite{ozkahraman2020combining}, and then repreating, the order of control policy executions cannot be described by a partial order and the prepares graph would contain cycles.
    % \emph{To the best of our knowledge, in this letter, we present the first BT stability theorem when there are cyclical behaviours. Moreover, we show that there is a hierarchy of BT stability properties (in order of inclusion): linear, acyclical, cyclical.}

    % \begin{figure}
    %     \small
    %     \begin{forest}
    %         for tree={
    %             minimum size=1.5em,
    %             inner sep=1pt
    %         }
    %         [{Sequential composition}, draw
    %             [{BTs}, draw
    %                 [{Sequences}, draw]
    %                 [{Implicit sequences}, draw]
    %                 [{Backchained BTs}, draw]
    %             ]
    %             [{Logical DSs}, draw]
    %             [{LQR trees}, draw]
    %         ]
    %     \end{forest}
    %     \caption{Taxonomy of sequential composition architectures.}
    % \end{figure}

    \section{Preliminaries}\label{sec:preliminaries}
    
    In this section, we will summarize the formal description of BTs presented in   \cite{sprague2021continuous}. As BTs can be seen as ordered trees we first review some key notation and concepts from that area, and then move on to the actual BTs.
    
    \subsection{Ordered trees}\label{sec:ordered-trees}
    A BT is structurally represented by an ordered tree \cite{kuboyama2007matching}, a rooted tree in which, in addition to
    an ordering between ancestors (parents, grandparents, and so on), there is a specified ordering between siblings (children of the same parent) \cite{sprague2021continuous}.
    
    One way to specify these orderings is by designating \textit{two} sets of edges in the directed-graph formalism.
    Given a finite set of vertices $V$, we would then specify a set of edges from children to parent $E_P \subset V^2$ 
    and from sibling to sibling $E_S \subset V^2$.
    If $(V, E_S)$ is a directed acyclical graph (DAG), $(V, E_P)$ is a rooted tree (also a DAG), and $E_P \cap E_S = \emptyset$ (parents cannot be siblings with children), then $(V, E_P, E_S)$ is an ordered tree, as illustrated in Fig.~\ref{fig:ordered_tree}. Note that the drawing of a tree $(V, E_P)$ on paper implicitly creates a sibling order (with e.g., siblings ordered from left to right), but this order is not present in $(V, E_P)$ itself.
    
        \begin{figure}[!h]
        \centering
        \begin{forest}
            for tree={
                minimum height=2em,
                minimum width=2em,
                s sep=6mm,
                edge={<-}
            }
            [A, draw
                [B, draw, name=B
                    [E, draw, name=E],
                    [F, draw, name=F]
                ],
                [C, draw, name=C
                    [G, draw, name=G],
                    [H, draw, name=H]
                ],
                [D, draw, name=D]
            ]
            \draw[double, ->] (B) -- (C);
            \draw[double, ->] (C) -- (D);
            \draw[double, ->] (E) -- (F);
            \draw[double, ->] (G) -- (H);
        \end{forest}
               \caption{An ordered tree written as two DAGs, with children-to-parent edges, $E_P$, indicated by single stroke lines and sibling-to-sibling edges, $E_S$, indicated by double-stroked lines.}
        \label{fig:ordered_tree}
    \end{figure}
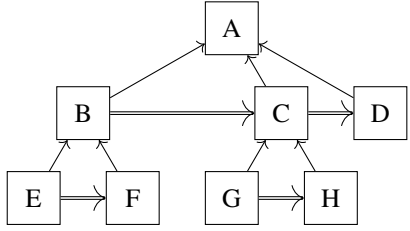

    Another way to specify these orderings can be created using \textit{partial orders}.
    Following \cite{sprague2021continuous}, a partial order can be obtained from the reflexive-transitive closure of a DAG's edges.
    Thus, by doing so on the edges, $E_P$ and $E_S$, we obtain a \textit{parent order} $\leq_P \subset V^2$ and a \textit{sibling order} $\leq_S \subset V^2$, respectively (also known as reachability relations on $E_P$ and $E_S$).
    An ordered tree could then be defined as $(V, \leq_P, \leq_S)$, a set of vertices, and two partial orders.

    Again, looking at the example in Fig.~\ref{fig:ordered_tree} we have that $B\leq_P E$ since there is an edge in $E_P$ from $E$ to $B$. Similarly, we have that $A\leq_P B$ and $A\leq_P C$ and so on. But we also have $A\leq_P E$ due to the transitive closure (see below), even though $(A,E) \not \in E_P$
  and   $A\leq_P A$ due to the reflexive closure (see below),  even though $(A,A) \not \in E_P$.
Similarly, for the sibling order we get $B\leq_S C \leq_S D$, $E\leq_S F$ and so on.
    
    Formally, a partial order $\leq$ on a set $V$ (such that $\leq \subset V^2$) is a homogeneous binary relation (if $(i,j) \in \leq$ we write $i \leq j$) that is 
    \textit{reflexive} ($\forall i \in V : i \leq i$), 
    \textit{antisymmetric} ($\forall i,j \in V : (i \leq j) \land (j \leq i) \implies i = j$), 
    and \textit{transitive} ($\forall i,j,k \in V: (i \leq j) \land (j \leq k) \implies i \leq k$).
    It is ``partial'' because there can exist pairs $(i, j) \in V^2$ for which neither $i \leq j$ nor $j \leq i$, in which case such pairs are \textit{incomparable}.
    If a pair is not incomparable, then it is \textit{comparable}.
    If all pairs in $V$ are comparable, then the partial order $\leq$ is more specifically a \textit{total order}.
    
    A subset of $V$ for which all pairs are comparable (resp. incomparable) with $\leq$ is a \textit{chain} (resp. \textit{antichain}) with respect to $\leq$.
    A chain (resp. antichain) is \textit{maximal} if it is not a proper subset of any other chain (resp. antichain).
    
       Thus, in Fig.~\ref{fig:ordered_tree} we have that $A, B, E$ is a maximal chain with respect to $\leq_P$, and $B, C$ is a chain (but not maximal) with respect to $\leq_S$. Furthermore, $B,G$ are not comparable, using either one of $\leq_P$ and  $\leq_S$.
    
    Where convenient, we will also use the corresponding \textit{converse relation} ($\geq$) and \textit{strict} orders ($<$ and $>$) of the relations above.
    If $i \leq j$, we can write $j \geq i$.
    If $i \leq j$ and $i \neq j$, we can write $i < j$ or $j > i$.
    
    The parent and sibling orders allow us to make the following distinctions.
    For any $i,j \in V$, if $i <_P j$ (resp. $i >_P j$), then $i$ is said to be an \textit{ancestor} of $j$ (resp. \textit{descendant}); if $i <_S j$ (resp. $i >_S j$), then $i$ is said to be a \textit{left-sibling} (resp. \textit{right-sibling}) of $j$.
    Their strict orders are \textit{complementary} to each other, meaning that if $(i,j)$ is comparable with $<_S$ (resp. $<_P$) then it is incomparable with $<_P$ (resp. $<_S$).
    
    Other orders can be created through \textit{compositions}.
    The composition of two relations is defined as $$\leq_A \circ \leq_B := \left\{(i,k)  \in V^2 \middle| \exists j : \left(i \leq_A j\right) \land \left(j \leq_B k\right)\right\}.$$
    
    Some composition orders that will be useful in the following are the
    left-uncle order $<_{LU} := <_S \circ \leq_P$,  the right-uncle order $>_{RU} := >_S \circ \leq_P$, the left-to-right order $<_{LR} := \geq_P \circ <_{LU}$, and the right-to-left order $>_{RL} := \geq_P \circ >_{RU}$, which are illustrated in Fig.~\ref{fig:uncle-orders}, and used
    e.g. Equation (\ref{eq:influence}) below.

    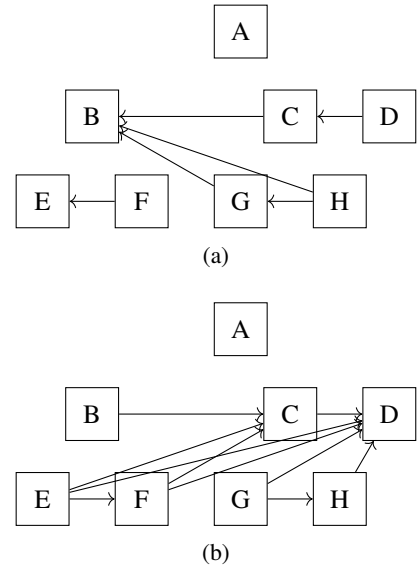
\begin{figure}[ht]
        \centering
        \begin{subfigure}[]{\columnwidth}
            \centering
            % \footnotesize
        \begin{forest}
            for tree={
                minimum height=2em,
                minimum width=2em,
                s sep=6mm,
                edge={<-}
            }
            [A, draw
                [B, draw, name=B, no edge
                    [E, draw, name=E, no edge],
                    [F, draw, name=F, no edge]
                ],
                [C, draw, name=C, no edge
                    [G, draw, name=G, no edge],
                    [H, draw, name=H, no edge]
                ],
                [D, draw, name=D, no edge]
            ]
            \draw[->] (C) -- (B);
            \draw[->] (D) -- (C);
            \draw[->] (F) -- (E);
            \draw[->] (H) -- (G);
            \draw[->] (G) -- (B);
            \draw[->] (H) -- (B);
        \end{forest}
            \caption{}
        \end{subfigure}\\
        \par\bigskip
        \begin{subfigure}{\columnwidth}
            \centering
            % \footnotesize
        \begin{forest}
            for tree={
                minimum height=2em,
                minimum width=2em,
                s sep=6mm,
                edge={<-}
            }
            [A, draw
                [B, draw, name=B, no edge
                    [E, draw, name=E, no edge],
                    [F, draw, name=F, no edge]
                ],
                [C, draw, name=C, no edge
                    [G, draw, name=G, no edge],
                    [H, draw, name=H, no edge]
                ],
                [D, draw, name=D, no edge]
            ]
            \draw[->] (B) -- (C);
            \draw[->] (C) -- (D);
            \draw[->] (E) -- (F);
            \draw[->] (G) -- (H);
            \draw[->] (E) -- (C);
            \draw[->] (E) -- (D);
            \draw[->] (F) -- (C);
            \draw[->] (F) -- (D);
            \draw[->] (G) -- (D);
            \draw[->] (H) -- (D);
        \end{forest}
        \caption{}
        \end{subfigure}
        \caption{Illustrations of the left-uncle order $<_{LU}$ (a) and the right-uncle order $>_{RU}$ (b), where the edges represent their reflexive-transitive reduction.}
        \label{fig:uncle-orders}
    \end{figure}
    
    These compositions yield strict partial orders;
    we define their corresponding partial orders as their reflexive closure, denoted by:
    $\leq_{LU}$, $\leq_{RU}$, $\leq_{LR}$, and $\leq_{RL}$. 
    
    Apart from the orders defined above, the parent map $p(i) := \sup_{\leq_P} \left\{j \middle| j <_P i\right\}$ will also be used.
    
       In Fig.~\ref{fig:uncle-orders} we illustrate the composite orders, given the ordered tree in  Fig.~\ref{fig:ordered_tree}. As can be seen,
       we have that $B <_{LU} H$, but since the non-strict $\leq_P$ is part of the definition of $<_{LU}$ we also have that 
   $G <_{LU} H$. Furthermore, regarding the left-to-right order we have that $E <_{LR} G <_{LR} H <_{LR} D$, but any pair that is comparable using $\leq_P$ are not comparable using $<_{LR}$, such as e.g. $A,C$.

    \subsection{Behavior Trees}\label{sec:bt-dynamical-system}
    
    % In this section, we will state the continuous-time formulation of BTs needed for the following section, following \cite{sprague2021continuous}.
    % Throughout this letter, let the state-space be $\mathbb{R}^n$ and the control-space be $\mathbb{R}^m$.
    We will now provide the formal description of BTs, following \cite{sprague2021continuous}.
    Throughout this paper, let $\mathbb{R}^n$ be the state space, $x \in \mathbb{R}^n$ be a state, and $\mathbb{R}^m$ be the control space with $u \in \mathbb{R}^m$.
    
    % \subsection{State-space formulation}\label{sec:statespace}
    
    \begin{definition}[Behavior Tree]\label{def:behavior}
        A function 
        $\bt_i: \mathbb{R}^n \to \mathbb{R}^m \times \{\mathcal{R},\mathcal{S},\mathcal{F}\}$
        defined as
        % defined for $x\in \mathbb{R}^n$ as
        \begin{equation}\label{eq:behavior}
            \bt_i\left(x\right) :=\left(u_i\left(x\right) , r_i\left(x\right) \right)
        \end{equation}
        where $i \in V$ is an index in an ordered tree,
        $u_i: \mathbb{R}^n \rightarrow  \mathbb{R}^m$ is a controller,
        and $r_i: \mathbb{R}^n \rightarrow  \{\mathcal{R},\mathcal{S},\mathcal{F}\}$
        is a metadata function, describing the progress of the controller in terms of the 
        outputs:
        running ($\mathcal{R}$),
        success ($\mathcal{S}$), 
        and
        failure ($\mathcal{F}$).
        % https://en.wikipedia.org/wiki/Function_(mathematics)#As_an_element_of_a_Cartesian_product_over_a_domain
        Define the metadata regions as the running, success, and failure regions:
        % \begin{equation}\label{eq:metadata-regions}
        %     \begin{aligned}
        %         R_i &:=\left\{x \in \mathbb{R}^n: r_i(x)=\mathcal{R} \right\} \\
        %         S_i &:=\left\{x \in \mathbb{R}^n: r_i(x)=\mathcal{S} \right\} \\
        %         F_i &:=\left\{x \in \mathbb{R}^n: r_i(x)=\mathcal{F} \right\},
        %     \end{aligned}
        % \end{equation}
        \begin{equation}\label{eq:metadata-regions}
        % R_i :=\left\{x: r_i(x)=\mathcal{R} \right\}, 
        % S_i :=\left\{x: r_i(x)=\mathcal{S} \right\},  F_i :=\left\{x: r_i(x)=\mathcal{F} \right\}
    \begin{gathered}
        R_i :=\left\{x: r_i(x)=\mathcal{R} \right\}, \\
        S_i :=\left\{x: r_i(x)=\mathcal{S} \right\}, \quad F_i :=\left\{x: r_i(x)=\mathcal{F} \right\},
    \end{gathered}
            % \begin{aligned}
            %     R_i &:=\left\{x: r_i(x)=\mathcal{R} \right\}, \\
            %     S_i &:=\left\{x: r_i(x)=\mathcal{S} \right\}, \quad F_i :=\left\{x: r_i(x)=\mathcal{F} \right\},
            % \end{aligned}
        \end{equation}
        respectively,
        which are pairwise disjoint and cover $\mathbb{R}^n$.
        % because $r_i$ is a well-defined total function.
       % Define the space of all possible BTs as
        %$\mathfrak{T} := (\mathbb{R}^n \times \{\mathcal{R},\mathcal{S},\mathcal{F}\})^{\mathbb{R}^n}$.
    \end{definition}
    
    % The requirement that the metadata regions (\ref{eq:metadata-regions}) are pairwise disjoint and cover $\mathbb{R}^n$ implies that the metadata function $r_i$
    % must be a \textit{well-defined} \textit{total function}.
    
    Note that in an ordered tree, every vertex $i \in V$ can be seen as the root of a subtree, consisting of all its descendants $\{j:i <_P j\}$. Thus, as above, $i \in V$ can also be seen as an index of a (sub)tree.
    As such, let us designate the \textit{root BT} as $\bt_0 = \bt_i$ (has index $0$) such that $\{j : j <_P i\} = \emptyset$,
    and \textit{sub BTs} as all other BTs $\bt_i$ such that $i \neq 0$.
    
    The intuition of the metadata regions (\ref{eq:metadata-regions}) is as follows.
    If $x \in S_i$, then $\bt_i$ has either succeeded in achieving its goal
    (e.g. opening a door) or the goal was already achieved (the door was already open).
    Either way, it might make sense to execute another BT in sequence to achieve another goal (perhaps a goal that is intended to be achieved after opening the door).
    
    If $x \in F_i$, then $\bt_i$ has either failed (the door to be opened turned out to be locked), or has no chance at succeeding (the door is impossible to reach from the current position).
    Either way, it might make sense to execute another BT as a fallback (either to open the door in another way or to achieve a higher-level goal in a way that does not involve opening the door).
    
    If $x \in R_i$, then it is too early to determine if $\bt_i$ will succeed or fail.
    In most cases, it makes sense to continue executing $\bt_i$, but it could also be reasonable to execute another BT if some other goal is more important (e.g. low battery levels indicate the need to recharge).
    
    Above, we use the term \say{execution} to mean the use of a BT's controller $u_i$ in some underlying \textit{autonomous dynamical system}, e.g. a robot.
    Thus, we have the following definition of a \textit{BT execution}.
    
    \begin{definition}[BT execution]\label{def:bt_execution}
        Given an autonomous dynamical system 
        $f : \mathbb{R}^n \times \mathbb{R}^m \to \mathbb{R}^n$ that is to be controlled, 
        the BT execution of $\bt_i$ is given in continuous-time as
        % the execution of a BT $\bt_i : i \in V$
        % and assuming the root of the BT is $\bt_0$ (has index 0), we have two cases, depending on if the system evolves in continuous- or discrete-time. 
        % For a continuous-time system we have
        \begin{equation}
            \label{eq:bt_execution}
            \dot x = f\left(x,u_i\left(x\right)\right),
        \end{equation}
        and for discrete-time as
        \begin{equation}
            \label{eq:bt_execution_discrete}
             x_{k + 1} = x_k + f\left(x_k,u_i\left(x_k\right)\right).
        \end{equation}
        In both cases, $u_i$ is given by (\ref{eq:behavior}).
    \end{definition}
    
    As shown in \cite{sprague2021continuous}, the continuous-time BT execution (\ref{eq:bt_execution})
    can be characterized by a \textit{discontinuous dynamical system} (DDS) defined over a finite set of so-called \textit{operating regions} \cite[Theorem 2, p.5]{sprague2021continuous}, with corresponding results regarding existence and uniqueness of its solutions \cite[Theorem 3, p.6]{sprague2021continuous}.
    These operating regions (with identical definitions for both continuous and discrete time dynamics found below) arise from the switching among BTs invoked by \textit{BT compositions}, of which there exist two fundamental types: \textit{Sequences} and \textit{Fallbacks}.
    
    A Sequence is a BT that composes together sub-BTs that are to be executed in a sequence where each one requires the success of the previous.
    It will succeed only if all sub-BTs succeed, whereas, if any one sub-BT fails, it will fail.
    This behavior is formalized in terms of the sub-BT's metadata regions in the following definition.
    
    \begin{definition}[Sequence]
        A functional
        % $Seq$ %: \bigcup_{m = 1}^\infty \mathfrak{T}^m \to \mathfrak{T}$
        that composes an arbitrarily finite sequence of 
        $M \in \mathbb{N}$ BTs
        into a new BT:
        \begin{equation}\label{eq:sequence}
            Seq\left[\bt_1, \dots, \bt_M\right]\left(x\right) :=
            \begin{cases}
            \bt_M\left(x\right) & \text{if} \quad x \in S_1 \cap \dots S_{M-1} \\
            \vdots & \vdots \\
            \bt_2\left(x\right) & \text{else-if} \quad  x \in S_1 \\
            \bt_1\left(x\right) & \text{else}.
            \end{cases}
        \end{equation}
        If $\bt_i = Seq[\bt_1, \dots, \bt_M]$,
        then $j,k \in \{1, \dots, M\}$
        are the children of $i$,
        such that $p(j) = i$,
        and are related as siblings, by $j \leq_S k$ if $j \leq k$.
        % and $\leq_S \supseteq = \leq \cap \{1, \dots, M\}^2$
        % and $$
        % The children are also ordered by $<_S$ such that
        % for $j,k\in \{1, \dots, M\}$, we have that 
        % $\bt_j <_S \bt_k$, if $j < k$.
        %Define
        %$\bt_i = Seq :\iff \bt_i(x) = Seq[(\bt_j)_{j \in \{j : p(j) = i\}}](x)$
        %in (\ref{eq:behavior-tree}).
    \end{definition}
    
    A Fallback, on the other hand, is a BT that composes together sub-BTs that are to be executed as a fallback to one another, where each one is only executed in the case of a failure of the previous.
    It will fail only if all sub-BTs fail, whereas, if any one sub-BT succeeds, it will succeed.
    This behavior is formalized in terms of the sub-BT's metadata regions in the following definition.
    
    \begin{definition}[Fallback]
        A functional
        % $Fal$ %: \bigcup_{m = 1}^\infty \mathfrak{T}^m \to \mathfrak{T}$
        that composes an arbitrarily finite sequence of 
        $M \in \mathbb{N}$ BTs
        into a new BT:
        \begin{equation}\label{eq:fallback}
            Fal\left[\bt_1, \dots, \bt_M\right]\left(x\right) :=
            \begin{cases}
            \bt_M\left(x\right) & \text{if} \quad x \in F_1 \cap \dots F_{M-1} \\
            \vdots & \vdots \\
            \bt_2\left(x\right) & \text{else-if} \quad  x \in F_1 \\
            \bt_1\left(x\right) & \text{else}.
            \end{cases}
        \end{equation}
        If $\bt_i = Fal[\bt_1, \dots, \bt_M]$,
        then $j,k \in \{1, \dots, M\}$
        are the children of $i$,
        such that $p(j) = i$,
        and are related as siblings, by $j \leq_S k$ if $j \leq k$.
        %which are totally ordered by $<_S$.
        %Define $\bt_i = Fal :\iff \bt_i(x) = Fal[(\bt_j)_{j \in \{j : p(j) = i\}}](x)$
        %in (\ref{eq:behavior-tree}).
    \end{definition}
    % As can be seen i
    
    \begin{lemma}\label{lemma:metadata-sequence}
        The metadata regions of a Sequence $\bt_i$ can be computed from the children's metadata regions as follows:
        
        % $R_i = \bigcup_{p\left(j\right) = i}\left( R_j \bigcap_{k <_S j} S_k\right)$,
        % $S_i = \bigcap_{p\left(j\right) = i} S_j$,
        % and $F_i = \bigcup_{p\left(j\right) = i}\left( F_j \bigcap_{k <_S j} S_k \right)$.
        % \begin{equation}\label{eq:sequence-metadata}
        %     \begin{aligned}
        %         R_i &= \bigcup_{p\left(j\right) = i}\left( R_j \bigcap_{k <_S j} S_k\right) \\
        %         S_i &= \bigcap_{p\left(j\right) = i} S_j, \quad \\
        %         F_i &= \bigcup_{p\left(j\right) = i}\left( F_j \bigcap_{k <_S j} S_k \right).
        %     \end{aligned}
        % \end{equation}
        \begin{equation}\label{eq:sequence-metadata}
        \begin{gathered}
            R_i =\bigcup_{p\left(j\right) = i}\left( R_j \bigcap_{k <_S j} S_k\right), \\
            S_i = \bigcap_{p\left(j\right) = i} S_j,\quad F_i = \bigcup_{p\left(j\right) = i}\left( F_j \bigcap_{k <_S j} S_k \right).
        \end{gathered}
        % ,~ S_i = \bigcap_{p\left(j\right) = i} S_j,~ F_i = \bigcup_{p\left(j\right) = i}\left( F_j \bigcap_{k <_S j} S_k \right)
            % \begin{array}{cc}
            %     R_i = \bigcup_{p\left(j\right) = i}\left( R_j \bigcap_{k <_S j} S_k\right) \\
            %     S_i = \bigcap_{p\left(j\right) = i} S_j, \quad 
            %     F_i = \bigcup_{p\left(j\right) = i}\left( F_j \bigcap_{k <_S j} S_k \right).
            % \end{array}
        \end{equation}
        % \begin{equation}\label{eq:sequence-metadata-regions}
        %     R_i = \bigcup_{p\left(j\right) = i}( R_j \bigcap_{k <_S j} S_k), ~
        %     S_i = \bigcap_{p\left(j\right) = i} S_j, ~
        %     F_i = \bigcup_{p\left(j\right) = i}( F_j \bigcap_{k <_S j} S_k ),
        % \end{equation}
        % which are pairwise disjoint and cover $\mathbb{R}^n$.
    \end{lemma}
    \begin{proof}
    See \cite[Lemma 1, p.4]{sprague2021continuous}.
    % If this is included in Letter, we should not repeat it here
%        A straightforward application of 
%        (\ref{eq:metadata-regions})
%        and (\ref{eq:sequence}). The running region of the sequence is the running region of the first child and the intersection of the success region of the first child with the running region of the second child and so on. The failure region works similarly, whereas the success region is the intersection of all the children success regions, as the sequence requires all children to succeed to return success.
    \end{proof}
    
    \begin{lemma}\label{lemma:metadata-fallback}
        % The metadata regions of a Fallback %($\bt_i = Fal$) 
        % are
        The metadata regions of a Fallback $\bt_i$ can be computed from the children metadata regions as follows
        % $R_i = \bigcup_{p\left(j\right) = i}\left( R_j \bigcap_{k <_S j} F_k\right)$,
        % $S_i = \bigcup_{p\left(j\right) = i}\left( S_j \bigcap_{k <_S j} F_k \right)$,
        % and $F_i = \bigcap_{p\left(j\right) = i} F_j$.
        % \begin{equation}\label{eq:fallback-metadata}
        %     \begin{aligned}
        %     R_i &= \bigcup_{p\left(j\right) = i}\left( R_j \bigcap_{k <_S j} F_k\right) \\
        %     S_i &= \bigcup_{p\left(j\right) = i}\left( S_j \bigcap_{k <_S j} F_k \right) \\
        %     F_i &= \bigcap_{p\left(j\right) = i} F_j.
        %     \end{aligned}
        % \end{equation}
            \begin{equation}\label{eq:fallback-metadata}
            \begin{gathered}
            R_i = \bigcup_{p\left(j\right) = i}\left( R_j \bigcap_{k <_S j} F_k\right), \\
            S_i = \bigcup_{p\left(j\right) = i}\left( S_j \bigcap_{k <_S j} F_k \right), \quad
            F_i = \bigcap_{p\left(j\right) = i} F_j.
            \end{gathered}
        \end{equation}
        % which are pairwise disjoint and cover $\mathbb{R}^n$
    \end{lemma}
    \begin{proof}
     See \cite[Lemma 2, p.4]{sprague2021continuous}.
    % If this is included in Letter, we should not repeat it here
%        A straightforward application of 
%        (\ref{eq:metadata-regions})
%        and (\ref{eq:fallback}). 
%        The running region is similar as for the Sequence above. The success region is similar to the running region, but the failure region is different since it requires all children to fail before returning failure.
    \end{proof}

    %new    
       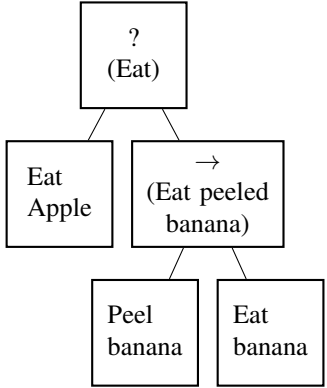
\begin{figure}[!h]
        \centering
        \begin{forest}
            for tree={
                minimum height=4em,
                minimum width=4em,
                inner sep=0.5pt,
            }
            [{$
            \begin{array}{c}
                ?\\
                \text{(Eat)}
            \end{array}
            $}, draw, thick
                [{$
                    \begin{array}{l}
                        \text{Eat}\\
                        \text{Apple}
                    \end{array}
                    $}, draw, thick],
                [{$
                \begin{array}{c}
                    \rightarrow\\
                    \text{(Eat peeled} \\
                    \text{banana)}
                \end{array}
                $}, draw, thick,
                    [{$
                    \begin{array}{l}
                        \text{Peel}\\
                        \text{banana}
                    \end{array}
                    $}, draw, thick],
                    [{$
                    \begin{array}{l}
                        \text{Eat}\\
                        \text{banana}
                    \end{array}
                    $}, draw, thick]
                ]
            ]
        \end{forest}
               \caption{A BT controlling the eating activities of an agent. Fallbacks are indicated by question marks and sequences by arrows. The parental ordering is indicated by lines. The sibling order is not  illustrated by arrows, but implicitly given by the ordering of the drawing. }
        \label{fig:eat_tree}
    \end{figure}
To illustrate the concepts above, we provide an example below, and to avoid getting into details about state space and implementation of individual control policies, we pick a high-level example that is hopefully familiar to most readers.
\begin{example}
Consider the eating policy of an agent as depicted in Fig. \ref{fig:eat_tree}. 
 The root node is a Fallback, indicated by the question mark. It is labeled Eat since it first tries to eat an apple, and if that fails tries to eat a banana. Thus it has two children: Eat apple and Eat banana. The latter in turn is a Sequence, indicated by an arrow. It is labeled Eat peeled banana, since it first tries to peel a banana and then, if that succeeds, tries to eat the banana. Thus it has two children: Peel banana and Eat banana.
 Eat peeled banana is a Sequence, since it only makes sense to progress to Eat banana if Peel banana succeeded. Eat on the other hand, is a fallback, since it (assuming only one fruit is needed) first tries to eat the apple, and only progresses to the banana if Eat apple fails.
 
 Now let's look at what Lemma \ref{lemma:metadata-sequence} and \ref{lemma:metadata-fallback} implies. First, we apply Lemma \ref{lemma:metadata-sequence} to Eat and peel banana (shortened EPB) which is a Sequence.
 From (\ref{eq:sequence-metadata}) we get,
 
\begin{align}
 R_{EPB} &=\bigcup_{p\left(j\right) = i}\left( R_j \bigcap_{k <_S j} S_k\right) = R_{peel} \cup (S_{peel}\cap R_{eat \hspace{0.5mm} banana}) \\
 S_{EPB} &= \bigcap_{p\left(j\right) = i} S_j =S_{peel} \cap S_{eat \hspace{0.5mm} banana},\\ 
 F_{EPB} &= \bigcup_{p\left(j\right) = i}\left( F_j \bigcap_{k <_S j} S_k \right) =F_{peel} \cup (S_{peel}\cap F_{eat \hspace{0.5mm} banana}).
\end{align}
Thus, the running region $R_{EPB}$ is the states where either the agent is currently peeling the banana, or states where the agent has succeeded in peeling and is currently eating. The success region $S_{EPB}$ includes states where the agent has succeeded in both peeling and eating the banana, and finally, the failure region  $F_{EPB}$ includes states where either the peeling failed, or the peeling succeeded, and the eating failed.

Given the above, we can look at the root of the tree, which is a Fallback.
When we apply Lemma \ref{lemma:metadata-fallback} and  Equation (\ref{eq:fallback-metadata})  to Eat 
 we get,
     \begin{align}
            R_{Eat} &= \bigcup_{p\left(j\right) = i}\left( R_j \bigcap_{k <_S j} F_k\right)=R_{eat \hspace{0.5mm} apple} \cup 
            (F_{eat \hspace{0.5mm} apple} \cap R_{EPB}), \\
            &= R_{eat \hspace{0.5mm} apple} \cup (F_{eat \hspace{0.5mm} apple} \cap (R_{peel} \cup (S_{peel}\cap R_{eat \hspace{0.5mm} banana}))), \nonumber\\
            S_{Eat} &= \bigcup_{p\left(j\right) = i}\left( S_j \bigcap_{k <_S j} F_k \right)=S_{eat \hspace{0.5mm} apple} \cup
             (F_{eat \hspace{0.5mm} apple} \cap S_{EPB}) , \\
             &= S_{eat \hspace{0.5mm} apple} \cup
             (F_{eat \hspace{0.5mm} apple} \cap (S_{peel} \cap S_{eat \hspace{0.5mm} banana})) , \nonumber\\
            F_{Eat} &= \bigcap_{p\left(j\right) = i} F_j = F_{eat \hspace{0.5mm} apple} \cap F_{EPB} \\
            &= F_{eat \hspace{0.5mm} apple} \cap (F_{peel} \cup (S_{peel}\cap F_{eat \hspace{0.5mm} banana}) \nonumber.
            \end{align}
            Thus, the running region $R_{Eat}$ includes states where the agent is either eating the apple, or failed in eating the apple and is currently peeling the banana, or failed in eating the apple, succeeded in peeling and is currently eating the banana.
            The success region $S_{Eat}$ includes states where the agent either succeeded in eating the apple, or failed in eating the apple, but succeeded in peeling and eating the banana.
            Finally, the failure region $F_{Eat}$ includes states where the agent failed with eating the apple and either failed to peel the banana, or succeeded in peeling but failed in eating the banana.
            
 \end{example}

    Now, we will define the aforementioned operating regions in terms of the BT compositions above.
    Operation regions $\Omega_i$ are subsets of the statespace for which $x \in \Omega_i$ is a \textit{sufficient condition}
    to conclude that $\bt_0(x) = \bt_i(x)$. That is, if the state is in the operating region of a sub-BT then that sub-BT is being executed in (\ref{eq:bt_execution}) or (\ref{eq:bt_execution_discrete}) for $\bt_0$.
    Within this sufficient condition, there is also a necessary condition characterized by the \textit{influence regions}.
    
    Influence regions $I_i \supseteq \Omega_i$ are subsets of the statespace for which $x \in I_i$ is a \textit{necessary condition}
    to conclude that $\bt_0(x) = \bt_i(x)$.
    Informally, they can be seen as the regions where the design of $\bt_i$ influences the execution, (\ref{eq:bt_execution}) or (\ref{eq:bt_execution_discrete}) for $\bt_0$,
    either by returning success so another BT is executed in sequence, by returning failure so another BT is executed as a fallback, or by itself being executed (returning running).
    Formally, the influence regions are defined in terms of the BT compositions, the strict left-uncle order, and the parent map, as follows.
    
    \begin{definition}[Influence Region] \label{def:influence}
        A subset of the statespace defined for $\bt_i$ as
        \begin{equation}\label{eq:influence}
            I_i := 
            \underset{\substack{j <_{LU} i \\ \bt_{p\left(j\right)} is \mbox{ } Seq }\quad}{\bigcap S_j}
            \cap
            \underset{\substack{j <_{LU} i \\ \bt_{p\left(j\right)} is \mbox{ } Fal.}\quad}{\bigcap F_j}
        \end{equation}
    \end{definition}
    
    As can be seen from (\ref{eq:influence}), if the state is in the influence region of a BT,
    then it is also in the success and failure regions of various other left-uncle BTs, but not their running regions.
    Thus, if the state is in both the influence region and running region of a BT, it is sufficient to conclude that the BT is being executed in (\ref{eq:bt_execution}) and returning running to the root.
    However, additional conditions are needed to conclude that a BT is being executed when it returns success or failure.
    
      %new
\begin{example}
 Looking at the same BT as above, depicted in Fig. \ref{fig:eat_tree}.  We have the following influence region for Eat banana.
  \begin{align}
  I_{eat \hspace{0.5mm} banana} &=  \underset{\substack{j <_{LU} i \\ \bt_{p\left(j\right)} is \mbox{ } Seq }\quad}{\bigcap S_j}
            \cap
            \underset{\substack{j <_{LU} i \\ \bt_{p\left(j\right)} is \mbox{ } Fal.}\quad}{\bigcap F_j} 
            = F_{eat \hspace{0.5mm} apple} \cap S_{peel}. 
            \end{align}
 Thus,  eating  the banana is only considered if eating the apple failed, and peeling the banana succeeded.
\end{example}

    A BT can return success (resp. failure) to the root if there 
    does not exist a right-sibling of either itself or any of its ancestors whose parent is a Sequence (resp. Fallback).
    The set of all sub-BTs that satisfy these conditions is given by the \textit{success and failure pathways} in terms of the BT compositions and strict right-uncle order, as follows.
    
    \begin{definition}[Success and failure pathways]
    \label{def:success_failure_pathways}
        % The terminal pathways are defined by the following index sets:
        % Subsets of the node $V$ define as
        Subsets of the vertices defined as
        \begin{align}
            \label{eq:success-pathway}
            \mathfrak{S} &:= \left\{i \in V \mid 
            \not \exists j \in V : \left(j >_{RU} i\right) \land 
            \left(\bt_{p(j)} ~ is ~ Seq\right)\right\} \\
            \label{eq:failure-pathway}
            \mathfrak{F} &:= \left\{i \in V \mid 
            \not \exists j \in V: \left(j >_{RU} i\right) \land 
            \left(\bt_{p(j)} ~ is ~ Fal\right)\right\},
        \end{align}
        respectively.
    \end{definition}
    
    As can be seen by (\ref{eq:success-pathway}) (resp. (\ref{eq:failure-pathway})), if a BT is in the success pathway (resp. failure pathway) then there does not exist another BT that will execute if it succeeds (resp. fails).
    Thus, if the state is in both the influence region and success region (resp. failure region) of a BT, and the BT is in the success pathway (resp. failure pathway), then it is sufficient to conclude that the BT is being executed in (\ref{eq:bt_execution}) or (\ref{eq:bt_execution_discrete}) for $\bt_0$.
    The operating regions are defined in terms of these sufficient conditions, as follows.

    \begin{definition}[Operating Region]
        A subset of the statespace
        defined for $\bt_i$ as
        \begin{equation}\label{eq:operating}
        \Omega_i := 
        \begin{cases}
            I_i \cap (R_i \cup S_i \cup F_i) = I_i & \text{if} \quad i \in \mathfrak{S} \cap \mathfrak{F} \\
            I_i \cap (R_i \cup S_i) & \text{else-if} \quad i \in \mathfrak{S} \\
            I_i \cap (R_i \cup F_i) & \text{else-if} \quad i \in \mathfrak{F} \\
            I_i \cap R_i  & \text{else}.
        \end{cases}
        \end{equation}
    \end{definition}
    
    With (\ref{eq:operating}), we know exactly when a BT will be executed in (\ref{eq:bt_execution}); namely if $x \in \Omega_i$ then 
    $\bt_0(x) = \bt_i(x)$, and hence
    $u_0(x) = u_i(x)$.
    We will use these operating regions to analyze the convergence of (\ref{eq:bt_execution}) or (\ref{eq:bt_execution_discrete}) for $\bt_0$ in the following section.
    
        %new
\begin{example}
 Looking once more at the same BT, depicted in Fig. \ref{fig:eat_tree}.  We have the following 
 
  \begin{align}
            % \label{eq:success-pathway}
            \mathfrak{S} &= \left\{i \in V \mid 
            \not \exists j \in V : \left(j >_{RU} i\right) \land 
            \left(\bt_{p(j)} ~ is ~ Seq\right)\right\} \\
            &= \{\mbox{Eat apple}, \mbox{Eat banana},  \mbox{Eat and peel banana},  \mbox{Eat}\} \nonumber \\
            % \label{eq:failure-pathway}
            \mathfrak{F} &= \left\{i \in V \mid 
            \not \exists j \in V: \left(j >_{RU} i\right) \land 
            \left(\bt_{p(j)} ~ is ~ Fal\right)\right\}, \\
            &= \{\mbox{Peel banana}, \mbox{Eat banana}, \mbox{Eat and peel banana},  \mbox{Eat}\} \nonumber
        \end{align}

  Thus,  if eating the apple or the banana succeeds, the entire BT will succeed, and the same holds for the non-leaves Eat and peel banana, and just Eat.
  Similarly,  if peeling or eating the banana fails, the entire BT will fail. Note that the last two are only executed when eat apple has already failed. The same holds for the non-leaves Eat and peel banana, and just Eat.
\end{example}

\section{Main result} \label{sec:main}

In this section, we will present a general BT convergence theorem (Section \ref{sec:convergence}) and show how it generalizes previous proofs (Section \ref{sec:generalisations}).
    
    \subsection{Convergence}\label{sec:convergence}

    In this section we will present a general convergence theorem for BTs, inspired by the idea of sequential composition in \cite{burridge1999sequential}.
    The theorem is based on the concept of a Domain of Attraction (DOA), a region of the statespace $B_i \subseteq \mathbb{R}^n$ 
    that is positively invariant with respect to the execution of a sub-BT, 
    (\ref{eq:bt_execution}) or (\ref{eq:bt_execution_discrete}) for $\bt_i$,
    % $\dot{x} = f(x, u_i(x))$ or
    % $x_{k + 1} = x_k + f(x_k,u_0(x_k))$ for $i \in V$,
    meaning that if the state starts in the region $x(0) \in B_i$, then the execution's maximal solution will stay in the region $x(t) \in B_i$.
    Additionally, we will assume below that if the state starts in the DOA of a sub-BT, the sub-BT's execution will make the state converge to some goal region $G_i \subseteq B_i$ in finite time.
    We will assume that all BTs have a DOA that is contained in the running and success regions and that the goal region is contained in the success region.
    We formalize this in the following definition and assumption.
    
    \begin{definition}[Finite-time successful]\label{def:fts}
        We call a BT $\bt_i$ finite-time successful if
        there exists a positively invariant region $B_i \subseteq R_i \cup S_i$ (called the Domain of Attraction, DOA),
        a positively invariant goal region $G_i \subseteq B_i \cap S_i$,
        and a finite time $\tau_i \in \mathbb{R}_{>0}$ such that
        $x(0) \in B_i$ implies $x(t') \in G_i$ in some smaller finite time $t' \in [0, \tau_i]$ (depending on the starting state $x(0)$)
        for solutions to %$\dot{x} = f(x, u_i(x))$.
        the BT execution in Def. \ref{def:bt_execution}. %(\ref{eq:bt_execution}) or (\ref{eq:bt_execution_discrete}).
    \end{definition}

    Note that the definition allows the DOA of a sub-BT to be empty ($B_i = \emptyset$), 
    thereby accommodating sub-BTs that do not inherently have a DOA.
    This could be the case, for example, in a sub-BT that uses a black-box control policy, such as a neural network (as in \cite{sprague2022adding}).
    
    One way to guarantee finite-time success, is via \textit{exponential stability}, following \cite{colledanchise2016behavior}.
    \begin{lemma}[Exponential stability]
        A BT $\bt_i$, for which $x_i \in G_i$ 
        is a locally exponentially stable equilibrium
        on $B_i$
        % for Def.~\ref{def:bt_execution}
        % for $\dot{x} = f(x, u_i(x))$
        for the execution (\ref{eq:bt_execution})
        is finite-time successful.
    \end{lemma}
    \begin{proof}
        The system 
        % $\dot{x} = f(x, u_i(x))$ in 
        (\ref{eq:bt_execution})  is locally exponentially stable about $x_i \in G_i$
        on $B_i \subseteq R_i \cup S_i$
        if there exists $\alpha_i, \beta_i \in \mathbb{R}_{>0}$
        such that, if $x(0) \in B_i$ then
        $\Vert x(t) - x_i \Vert \leq  \Vert x(0) - x_i \Vert \alpha_i e^{- \beta_i t}$
        for all $t \in \mathbb{R}_{\geq 0}$.
        % \begin{equation}\label{eq:exponential}
        %     \Vert x(t) - x_i \Vert \leq  \Vert x(0) - x_i \Vert \alpha_i e^{- \beta_i t} \quad \forall t \in \mathbb{R}_{\geq 0}.
        % \end{equation}
        The stability of $x_i \in B_i \cap S_i$ implies that there must exist
        a maximal $\epsilon_i \in \mathbb{R}_{>0}$ and minimal $\tau_i \in \mathbb{R}_{>0}$
        such that
        $\Vert x(\tau_i) - x_i \Vert \leq  \Vert x(0) - x_i \Vert \alpha_i e^{- \beta_i \tau_i} < \epsilon_i$
        and $\{x \in \mathbb{R}^n : \Vert x - x_i \Vert \leq \epsilon_i\} \subseteq S_i$.
        % Thus, we must have (\ref{eq:fts}).
        A similar argument holds for the discrete-time case of  (\ref{eq:bt_execution_discrete}).
    \end{proof}

\begin{figure}[t]
    \centering
    \includegraphics[width=\linewidth]{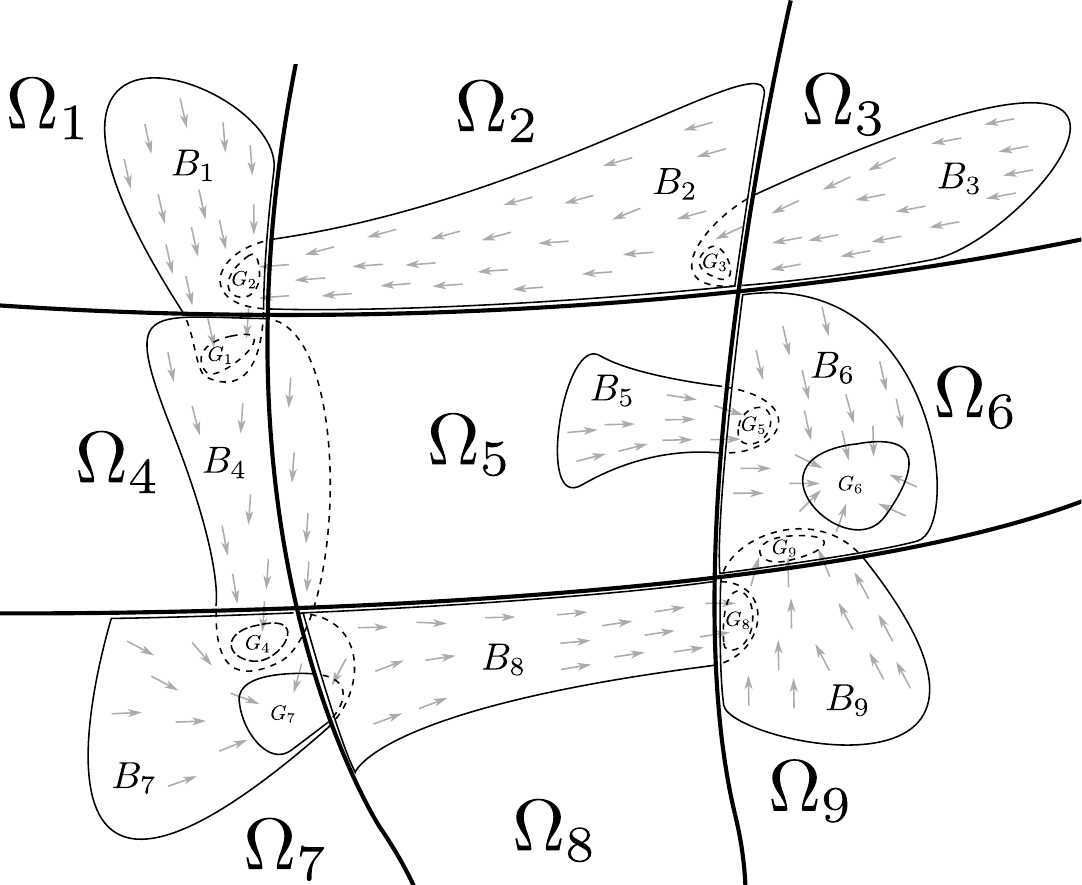}
    \caption{A copy of Fig.~\ref{fig:operating_regions} provided for convenience. An example of how the operating regions $\Omega_i$ and domains of attraction $B_i$ can be arranged to make the state reach one of the two desired goal regions $G_6, G_7$ from a wide variety of starting states. Outside the domains of attraction $B_i$, we have no guarantees on solution behavior.}
    \label{fig:operating_regions_v2}
\end{figure}
    
To provide intuition for the main result (Theorem \ref{theorem:general_convergence}), and its proof, we consider the example of Fig.~\ref{fig:operating_regions} that is repeated in Fig.~\ref{fig:operating_regions_v2} for convenience. Looking at  Fig.~\ref{fig:operating_regions_v2} we see that the vector fields indicate that transitions from $B_3$ to $B_2$ are likely. Similarly, transitions are likely from $B_2$ to $B_1$ to $B_4$ to $B_7$. From $B_7$ it might either go to the goal region $G_7$, or possibly to $B_8$ and then $B_9$ and $B_6$, followed by the goal at $G_6$. Finally, if it starts in $B_5$ it will also move to $B_6$. There is however a potential problem: $B_4$ extends into $\Omega_5$, so the execution of $u_4$ might lead to the execution of $u_5$ outside of $B_5$ and then we do not know what will happen.

To analyze the system we will partition every $\Omega_i$ into tree sets $\Omega_i=v_a(i) \cup v_b(i) \cup v_c(i)$, using the $B_i$ and $G_i$ from Definition~\ref{def:fts}. The first set is all states outside the DOA, $v_a(i)=\Omega_i \setminus B_i$. 
The second is the states in the DOA, but outside the goal, $v_b(i)=\Omega_i  \cap (B_i \setminus G_i)$, and the third set is the states in the goal, which is inside the DOA, $v_c(i)=\Omega_i \cap G_i$.

Given the above we can now create the prepares graph $\Gamma$
(see Section~\ref{sec:related-work}) shown in Fig.~\ref{fig:prepares_graph}. 
Each of the vertices of the graph represents one of the sets mentioned above, as indicated in the figure.
% To separate the vertex from its corresponding set, we use upper case notation for the vertices, such that $V_a(i)$ is the vertex representing the set $v_a(i)$,  $V_b(i)$ represents $v_b(i)$ and $V_c(i)$ represents $v_c(i)$.
 This graph is formally defined in Definition~\ref{def:prepares_graph} below, but intuitively we note that the edges correspond to possible transitions between the sets. We have no knowledge of the vector field outside $B_i$, so the possible transitions out of the set $v_a(i)$ are given by all neighboring sets, so the vertex $v_a(2)$ has edges to $v_a(1), v_a(3), v_b(1), v_b(2)$.
The potential problem with $B_4$ identified above can be seen in terms of a bidirectional edge between $v_b(4)$ and $v_a(5)$.

    \begin{figure}[t]
        \centering
        \begin{tikzpicture}[node distance={15mm}, main/.style = {draw}]
        % DOAs
        \node[main, very thick] (1) [] {$v_b(1)$};
        \node[main, very thick] (2) [right of=1] {$v_b(2)$};
        \node[main, very thick] (3) [right of=2] {$v_b(3)$};
        \node[main, very thick] (4) [below of=1] {$v_b(4)$};
        \node[main, very thick] (18) [right of=4] {$v_a(5)$};
        \node[main, very thick] (5) [right of=18] {$v_b(5)$};
        \node[main, very thick] (6) [right of=5] {$v_b(6)$};
        \node[main, very thick] (7) [below of=4] {$v_b(7)$};
        \node[main, very thick] (8) [right of=7] {$v_b(8)$};
        \node[main, very thick] (9) [right of=8] {$v_b(9)$};
        % non-DOA
        \node[main] (10) [label=above left:{$\Gamma$}, above left of=1] {$v_a(1)$};
        \node[main] (11) [above of=2] {$v_a(2)$};
        \node[main] (12) [above right of=3] {$v_a(3)$};
        \node[main] (13) [left of=4] {$v_a(4)$};
        \node[main] (17) [right of=6] {$v_a(6)$};
        % goal
        \node[main, very thick] (19) [below of=6] {$v_c(6)$};
        \node[main] (16) [below right of=19] {$v_a(9)$};
        \node[main, very thick] (20) [below right of=7] {$v_c(7)$};
        \node[main] (14) [below left of=20] {$v_a(7)$};
        \node[main] (15) [below right of=20] {$v_a(8)$};
        
        % prepares edges
        \draw[->] (2) -- (1);
        \draw[->] (3) -- (2);
        \draw[->] (1) -- (4);
        \draw[->] (4) -- (7);
        \draw[->] (7) -- (8);
        \draw[->] (8) -- (9);
        \draw[->] (9) -- (6);
        \draw[<->] (4) -- (18);
        \draw[->] (18) -- (5);
        \draw[->] (5) -- (6);
        \draw[->] (5) -- (6);
        \draw[->] (18) -- (2);
        \draw[->] (18) -- (8);
        \draw[->] (9) -- (6);
        \draw[->] (17) -- (6);
        \draw[->] (6) -- (19);
        \draw[<->] (10) -- (11);
        \draw[<->] (11) -- (12);
        \draw[<->] (12) -- (17);
        \draw[<->] (10) -- (13);
        \draw[<->] (13) -- (14);
        \draw[<->] (14) -- (15);
        \draw[<->] (15) -- (16);
        \draw[<->] (16) -- (17);
        \draw[->] (7) -- (20);
        \draw[->] (10) -- (1);
        \draw[->] (11) -- (2);
        \draw[->] (12) -- (3);
        \draw[->] (13) -- (4);
        \draw[->] (15) -- (8);
        \draw[->] (16) -- (9);
        \draw[->] (14) -- (7);
        \draw[->] (18) to [out=25, in=155] (6);
       
       % I think these are also needed /P
        \draw[->] (11) -- (1);
        \draw[->] (17) -- (3);
        \draw[->] (16) -- (6);
        \draw[->] (15) -- (9);
        \draw[->] (14) to [out=85, in=205] (8);

        \end{tikzpicture}
        \caption{The prepares graph $\Gamma$ corresponding to the operating regions and DOAs in Fig. \ref{fig:operating_regions}.
        The bolded vertices represent the analysis set $V_\Gamma'$.}
        \label{fig:prepares_graph}
    \end{figure}
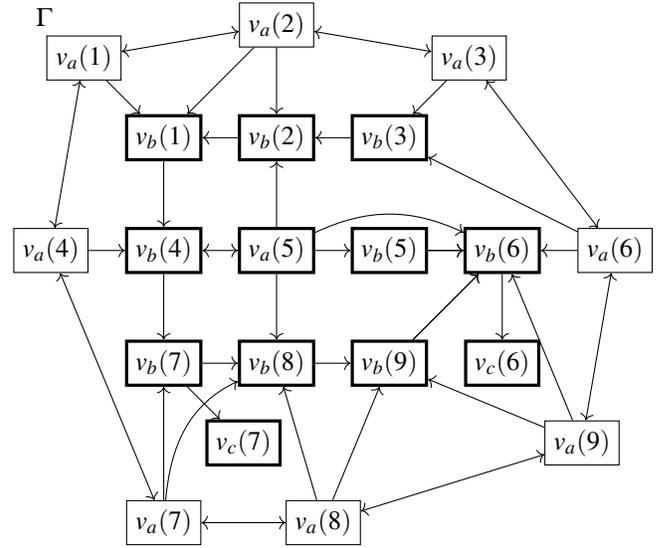    
    
    If there are no cycles in the prepares graph we can draw some conclusions regarding convergence, in a way that is similar to our discussion on the sets $B_i$ above. If there are cycles however, we convert the graph to a so-called condensed prepares graph $\Gamma_\star$ (see Definition \ref{def:condensed} below) where we merge all vertices that are mutually reachable. Starting from Fig. \ref{fig:prepares_graph} we get Fig. \ref{fig:condensation}. We see two merged vertices: $\{v_b(1), v_b(4), v_a(5), v_b(2)\}$ and $\{v_a(1), v_a(2), v_a(3), v_a(4), v_a(6), v_a(7), v_a(8), v_a(9)\}$.
    As can be seen, starting from e.g., $v_a(1)$ all other vertices in the merged vertex can be reached in a series of transitions.
    Any cycle is by definition mutually reachable, so the condensed prepares graph $\Gamma_\star$ is guaranteed to have no cycles.
    The formal definition of $\Gamma_\star$ can be found in Definition~\ref{def:condensed}.

     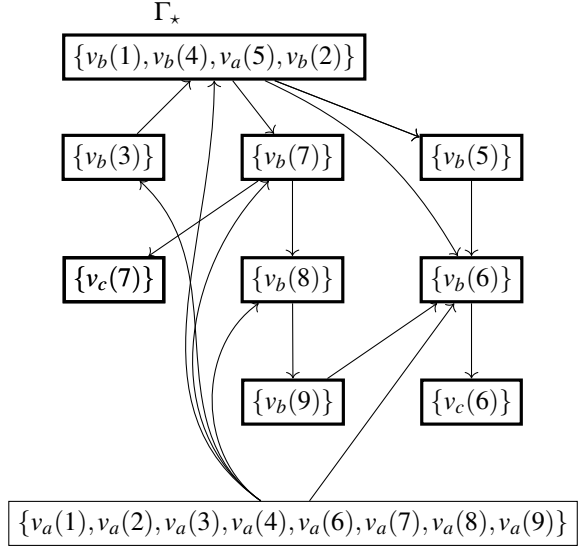
\begin{figure}[t]
        \centering
        \begin{tikzpicture}[main/.style = {draw}]
        % DOAs
        \node[main, very thick] (1) [label=above left:{$\Gamma_\star$}] {$\{v_b(1), v_b(4), v_a(5), v_b(2)\}$};
        \node[main, very thick] (2) [below right=1cm of 1] {$\{v_b(5)\}$};
        \node[main, very thick] (3) [below=1cm of 2] {$\{v_b(6)\}$};
        \node[main, very thick] (4) [below=1cm of 3] {$\{v_c(6)\}$};
        \node[main, very thick] (5) [left=1cm of 2] {$\{v_b(7)\}$};
        \node[main, very thick] (6) [left=1cm of 3] {$\{v_b(8)\}$};
        \node[main, very thick] (7) [left=1cm of 4] {$\{v_b(9)\}$};
        \node[main, very thick] (8) [left=1cm of 6] {$\{v_c(7)\}$};
        \node[main, very thick] (9) [left=1cm of 6] {$\{v_c(7)\}$};
        \node[main] (10) [below=1cm of 7] {$\{v_a(1), v_a(2), v_a(3), v_a(4), v_a(6), v_a(7), v_a(8), v_a(9)\}$};
        \node[main, very thick] (11) [left=1cm of 5] {$\{v_b(3)\}$};
        \draw[->] (1) -- (2);
        \draw[->] (2) -- (3);
        \draw[->] (3) -- (4);
        \draw[->] (1) -- (5);
        \draw[->] (5) -- (6);
        \draw[->] (6) -- (7);
        \draw[->] (7) -- (3);
        \draw[->] (1) -- (2);
        \draw[->] (5) -- (8);
        \draw[->] (11) -- (1);
        % \draw[->] (1) -- (3);
        \draw[->] (1) to [out=-25, in=120] (3);
        % outside in
        % \draw (1) to [out=135,in=90,looseness=1.5] (5);
        \draw[->] (10) to [out=145, in=270] (1);
        % \draw[->] (10) to [out=40, in=220] (2);
        \draw[->] (10) -- (3);
        \draw[->] (10) to [out=145, in=318] (11);
        \draw[->] (10) to [out=145, in=225] (5);
        \draw[->] (10) to [out=145, in=215] (6);
        \end{tikzpicture}
        \caption{The condensed prepares graph $\Gamma_\star$ corresponding to the operating regions and DOAs in Fig. \ref{fig:operating_regions}. The bolded vertices represent a possible choice of analysis set $V_\Gamma'$.}
        \label{fig:condensation}
    \end{figure}
    
      Now, looking at $\Gamma_\star=(V_\star, E_\star)$ in Fig.~\ref{fig:condensation}, 
      we can draw conclusions regarding convergence.
      First, we note that the only sinks (vertices without outgoing edges) are the two goal vertices, $\{v_c(6)\}, \{v_c(7)\}$. By definition, we will never leave those. Now, if we pick some subset of the vertices $V'_\star \subseteq V_\star$ such that there are no outgoing edges from $V_\star'$, and all vertices in $V_\star'$, except the goal vertices, are such that we will leave them in finite time, we know that any state starting in $\bigcup \bigcup V_\star'$ (the union of all the sets inside the union of some collections of sets) will end up in a goal state in finite time. 
      
      Again, looking at Fig.~\ref{fig:condensation}, we see that including several of the $v_b(i)$ vertices in $V_\star'$ is probably a good idea, since they correspond to DOAs that are likely to lead to a transition. We have less knowledge of merged vertices containing cycles, and/or $v_a(i)$, but sometimes we know that a cycle will not execute forever, for example when a lawnmower robot  goes to recharge regularly (creating a cycle) but will eventually have finished the whole lawn, or when a data-driven controller is added to an existing design, such as in Section~\ref{sec:NN_MB}. In such cases, a merged vertex can also be included in $V_\star'$.
    %   As above, we let lower case $v$ correspond to  sets of states and upper case $V$ be  vertices,  
    %   \begin{equation}
    %       v_L = \bigcup_{i \in V_L} v_i.
    %   \end{equation}
    %   Thus, if we choose $V_\star' =\{\{v_b(7)\} , \{v_c(7)\}, \{v_b(8)\}, \{v_b(9)\}, \{v_b(5)\}, \{v_b(6)\}, \{v_c(6)\}\}\subset V_\star$
    %   we get 
    %   $v_L=v_b(7) \cup v_c(7)\cup v_b(8)\cup v_b(9)\cup v_b(5)\cup v_b(6)\cup v_c(6)\subset \mathbb{R}^n$.
      Without any knowledge of the cycles of the merged vertices, we can at least conclude that $V_\star'$ is such that any state starting in the set $\bigcup \bigcup V_\star'$ 
    %   (the union of all the sets inside the union some collections of sets) 
      will reach one of the goal sets $v_c(6),v_c(7)$ in finite time.

      We will now present formal results, corresponding to the informal discussion above. But, we must first pick a set of  nodes in the BT such that their operating regions $\Omega_i$ together cover the entire state space, as illustrated in Fig.~\ref{fig:operating_regions_v2}.
    
    \begin{definition}[Abstraction]\label{def:abstraction}
         An abstraction of a BT is a subset of the vertices $P \subseteq V$ such that $\{\Omega_i\}_{i \in P}$ is a partition of the statespace $\mathbb{R}^n$.
         That is $\Omega_i \cap \Omega_j = \emptyset$ and $\bigcup_{i \in P} \Omega_i = \mathbb{R}^n$.
    \end{definition}
    Note that we cannot pick all nodes of the BT, as each child node has an operating region that is a subset of its parent's (see \cite[Lemma 3, p.5]{sprague2021continuous}). However, many options remain. All the leaves of a BT is a valid abstraction, and so is the single root node.
      
     \begin{assumption}[Finite-time successful]\label{ass:fts}
        All  BTs $\bt_i$ rooted in the vertices of $P$ are finite-time successful.
    \end{assumption}  
      
    In the following we need to handle both the cases of continuous- and discrete-time executions, corresponding to (\ref{eq:bt_execution}) and (\ref{eq:bt_execution_discrete}). Therefore we make the following definition of when subsets of the statespace  are close enough   to enable a transition of the state from one set to the other without entering any other sets.
      \begin{definition}[Neighboring sets]
      \label{def:neighboring_sets}
          The definition is a bit different for continuous-time and discrete-time executions.
          In the case of continuous-time execution (\ref{eq:bt_execution}) we say that two sets $\Sigma_0,\Sigma_1 \subset \mathbb{R}^n$ are neighboring, if $\partial \Sigma_0 \cap \partial \Sigma_1 \neq \emptyset$.
          
          In the case of discrete-time execution (\ref{eq:bt_execution_discrete}), we let
          $\delta = \sup_{x,u} \|f(x,u)\|$, i.e. an upper bound on the size of the state transition in a single time step.
          Then we say that two sets $\Sigma_0,\Sigma_1\subset X$ are neighboring, if $\inf_{\sigma_0 \in \Sigma_0,\sigma_1 \in \Sigma_1}  \|\sigma_0 - \sigma_1\|\leq \delta$.
      \end{definition}
      
      We are now ready to define the prepares graph, with vertices $v_a(i),v_b(i),v_c(i)$ and edges where transitions are possible.
      
     \begin{definition}[Prepares graph]\label{def:prepares_graph}
        Let 
        \begin{align}
            v_a(i)&=\Omega_i \setminus B_i \label{eq:va}\\
            v_b(i)&=\Omega_i \cap (B_i \setminus G_i) \label{eq:vb}\\
            v_c(i)&=\Omega_i \cap G_i \label{eq:vc}
        \end{align}
        for $i \in P$.
        % Let furthermore each of these sets be represented by a vertex $V_a(i),V_b(i),V_c(i)$, which together make up the vertices of the prepares graph
        By slight abuse of notation, let these sets be vertices of the prepares graph
        \begin{align}
            % V_\Gamma &=\{V_a(i)\}_{i \in P} \cup \{V_b(i)\}_{i \in P} \cup \{V_c(i)\}_{i \in P} \\
            \Gamma &:= (V_\Gamma, E_\Gamma) \\
            V_\Gamma &:= \bigcup_{i \in P} \left\{v_a(i), v_b(i), v_c(i)\right\},
        \end{align}
        where $E_\Gamma$ is a set of directed edges given by
        \begin{equation} \label{eq:prepares_edges}
            \begin{aligned}
                E_\Gamma :=
                &\left\{
                    \left(v_a(i), v_a(j)\right) \middle|
                    \begin{aligned}
                       v_a(i),v_a(j) \mbox{ are neighboring}& \\
                        %\partial v_a(i) \cap \partial v_a(j) &\neq \emptyset \\
                        \land  i \neq j  \land \left(i,j\right) \in P^2 &
                    \end{aligned}
                \right\} \cup \\
                &\left\{
                    \left(v_a(i), v_b(j)\right) \middle|
                    \begin{aligned}
                        v_a(i),v_b(j) \mbox{ are neighboring}& \\
                        %\partial v_a(i) \cap \partial v_b(j) &\neq \emptyset \\
                        \land \left(i,j\right) \in P^2&
                    \end{aligned}
                \right\} \cup\\
                &\left\{
                    \left(v_a(i), v_c(j)\right) \middle|
                    \begin{aligned}
                        v_a(i),v_c(j) \mbox{ are neighboring}& \\
                        %\partial v_a(i) \cap \partial v_c(j) &\neq \emptyset \\
                        \land \left(i,j\right) \in P^2 &
                    \end{aligned}
                \right\} \cup\\
                &\left\{
                    \left(v_b(i), v_a(j)\right) \middle|
                    \begin{aligned}
                        v_b(i),v_a(j) \mbox{ are neighboring}& \\
                        %\partial v_b(i) \cap \partial v_a(j) &\neq \emptyset \\
                        \land B_i  \cap v_a(j) \neq \emptyset &\\
                        \land i \neq j 
                        \land \left(i,j\right)\in P^2 &
                    \end{aligned}
                \right\} \cup \\
                &\left\{
                    \left(v_b(i), v_b(j)\right) \middle|
                    \begin{aligned}
                        v_b(i),v_b(j) \mbox{ are neighboring}& \\
                        %\partial v_b(i) \cap \partial v_b(j) &\neq \emptyset \\
                        \land B_i  \cap v_b(j) \neq \emptyset &\\
                        \land i \neq j 
                        \land \left(i,j\right) \in P^2 &
                    \end{aligned}
                \right\} \cup \\
                &\left\{
                    \left(v_b(i), v_c(j)\right) \middle|
                    \begin{aligned}
                    v_b(i),v_c(j) \mbox{ are neighboring}& \\
                       % \partial v_b(i) \cap \partial v_c(j) \neq \emptyset & \\
                        \land B_i  \cap v_c(j) \neq \emptyset&\\
                        % \land i \neq j 
                        \land \left(i,j\right) \in P^2 &
                    \end{aligned}
                \right\}.
                % &\left\{
                %     \left(v_b(i), v_c(i)\right) \middle| \Omega_i \cap G_i \neq \emptyset \land i \in P
                % \right\},
            \end{aligned}
        \end{equation}
        %for continuous time,
        %where $\partial$ denotes the boundary of a set.
        Furthermore, let $\leq_\Gamma$ be the reachability relation of the prepares graph $\Gamma$ (given by the reflexive-transitive closure of $E_\Gamma$).
        %Define the mapping $v := v_a^{-1} \cup v \cup v_c^{-1}$, a surjection from $V_a \cup V_b \cup V_c$ to $P$. 
        % Why is this needed?
    \end{definition}
    
    Note that we have edges going to all neighbors of the $v_a(i)$ vertices since we know nothing of the vector fields there. But for $v_b(i)$ we know that the set $B_i$ is invariant, so for those outgoing edges we also need $B_i$ to overlap with the neighboring set.
    
    As noted above $\Gamma$ may contain cycles, and then we want to create another graph $\Gamma_\star$ where all vertices in a cycle in $\Gamma$ are merged into a single vertex, leaving $\Gamma_\star$ without cycles.
    
    To do this we  use the reachability relation $\leq_\Gamma$ induced by the vertices of $\Gamma$ to define the following equivalence relation
    $\leq_\Gamma \cap \geq_\Gamma=\{(i,j) \in V_\Gamma^2 | (i \leq_\Gamma j) \land (i \geq_\Gamma j)\}$, of vertices that are mutually reachable, and thus parts of cycles.

    \begin{definition}[Condensed prepares graph]
    \label{def:condensed}
        The condensed prepares graph is a DAG  
        \begin{align}
            \Gamma_\star :=& \left(V_\star, E_\star\right) \\
            V_\star :=& V_\Gamma / (\leq_\Gamma \cap \geq_\Gamma),
        \end{align}
        % $\Gamma_\star := (V_\star, E_\star)$,
        % with 
        % $V_\star= V_\Gamma / (\leq_\Gamma \cap \geq_\Gamma)$,
        where $/$ is the modulo operator creating the equivalence classes and the set of edges are given by
        % Let  $V_\star(i)$ denote the elements of $V_\star$.
        % Furthermore, let the edges be given by
        \begin{equation}
        \label{eq:gamma_star_edges}
            E_\star := 
            \left\{\left(i,j\right) \in V_\star^2 \middle| 
            E_\Gamma \cap \left(i \times j\right) \neq \emptyset
            \right\},
            % \left\{\left(i,j\right) \in V_\star^2 \middle| 
            %  \exists \left(k,l\right) \in E_\Gamma: k \in V_\star(i) \land l \in V_\star(j)
            % \right\}.
        \end{equation}
        where $i, j \subseteq V_\Gamma$.
        Thus, there is an edge in $V_\star$ between two equivalence classes, if there is an edge in $V_\Gamma$ between any elements of the two classes.
        
        Again, we define the reachability relation of $\Gamma_\star$ as $\leq_{\Gamma_\star}$ (the reflexive-transitive closure of $E_\star$).
        Define the maximal strongly connected components as
        \begin{equation}
            \overline{V}_\star := \left\{i \in V_\star \middle| \not \exists j \in V_\star : i <_{\Gamma_\star} j\right\},
        \end{equation}
        i.e. the sinks of $\Gamma_\star$.
    \end{definition}

 As noted above,  an example of the $\Gamma_\star$ corresponding to the $\Gamma$ of Fig. \ref{fig:prepares_graph} can be found in Fig. \ref{fig:condensation}.
 
     The prepares graph $\Gamma$ and the condensed prepares graph $\Gamma_\star$ both yield information about what kind of transitions can happen when starting anywhere in the statespace.
    However, we are often only interested in what transitions can happen when starting from a particular region of the state space such as the DOA of a sub-BT.
    For example, as indicated in Fig.~\ref{fig:operating_regions}, if the state starts in $\Omega_3 \cap B_3$ (corresponding to $v_b(3)$ in $\Gamma$ and $\{v_b(3)\}$ in $\Gamma_\star$), then the state can transition to any of the regions corresponding to the bolded vertices in Figs. \ref{fig:prepares_graph} and \ref{fig:condensation}.
    We call these vertices the \say{analysis sets},
    $V_\Gamma'$ for $\Gamma$ and $V_\star'$ for $\Gamma_\star$, i.e. subsets of $\Gamma$'s and $\Gamma_\star$'s vertices, respectively, such they do not have any outgoing vertices.
    We formally define these sets below.
    
    \begin{definition}[Analysis sets]
        A subset of 
        $\Gamma$'s vertices
        \begin{equation}
            V_\Gamma' :\subseteq V_\Gamma \quad \text{s.t.} \quad E_\Gamma \cap \left(V_\Gamma' \times \left(V_\Gamma \setminus V_\Gamma'\right)\right) = \emptyset
        \end{equation}
        % $V_\Gamma' :\subseteq V_\Gamma$ such that
        % there does not exist $(i,j) \in E$ such that $i \in L$ and $j \not \in L$,
        % $L = \bigcup_{i \in L} \{j \in V_a \cup V_b \cup V_c | i \leq_\Gamma j\}$,
        and a subset of $\Gamma_\star$'s vertices
        % $L_\star := v_\star[L / (\leq_\Gamma \cap \geq_\Gamma)] \subseteq V_\star$,
        \begin{equation}
             V_\star' := V_\Gamma' / \left(\leq_\Gamma \cap \geq_\Gamma\right).
        \end{equation}
        Note that $V_\Gamma' = \bigcup V_\star'$.
        % where $[\cdot]$ denotes the image.
    \end{definition}

    So far, we have divided the operation regions $\Omega_i$ into subsets $v_a(i),v_b(i),v_c(i)$ and seen what transitions between these are possible, and we will pay extra attention to a chosen analysis set $V_\Gamma'$. Now we map these transitions back to the larger sets $\Omega_i$ to see what transitions are possible on that level.
    
    To do this we define a directed graph
    $\Gamma_\Omega=(V_\Omega, E_\Omega)$,
    where $V_\Omega$ is a subset of operating regions in $P$ corresponding to the analysis set $V_\Gamma'$.
    Further, the edges $E_\Omega$ are the ones where an edge exists in $V_\Gamma'$.
    This is illustrated in Fig.~\ref{fig:behaviour_graph} and formalized below.
    
    %Thus, using the analysis sets we  derive a directed graph $\Gamma_\Omega$ representing the possible transitions between different operating regions $\Omega_i$ by partitioning vertices in $\Gamma$ by the operating region they belong to, as illustrated in Fig. \ref{fig:behaviour_graph}.
    %We formalise this in the following definition. 
    
          \begin{figure}
        \centering
        \begin{tikzpicture}[main/.style = {draw}]
            \node[main, very thick] (1) [] {$1$};
            \node[main, very thick] (2) [right=1cm of 1] {$2$};
            \node[main, very thick] (3) [right=1cm of 2] {$3$};
            \node[main, very thick] (4) [below=1cm of 1] {$4$};
            \node[main, very thick] (5) [right=1cm of 4] {$5$};
            \node[main, very thick] (6) [right=1cm of 5] {$6$};
            \node[main, very thick] (7) [below=1cm of 4] {$7$};
            \node[main, very thick] (8) [right=1cm of 7] {$8$};
            \node[main, very thick] (9) [right=1cm of 8] {$9$};
            \draw[->] (3) -- (2);
            \draw[->] (2) -- (1);
            \draw[->] (1) -- (4);
            \draw[->] (4) -- (7);
            \draw[->] (7) -- (8);
            \draw[->] (8) -- (9);
            \draw[->] (9) -- (6);
            \draw[->] (5) -- (2);
            \draw[->] (5) -- (8);
            \draw[->] (5) -- (6);
            \draw[<->] (4) -- (5);
        \end{tikzpicture}
        \caption{The behavior graph $\Gamma_\Omega$ corresponding to the analysis sets, $V_\Gamma'$ and $V_\star'$, bolded in Figs. \ref{fig:prepares_graph} and \ref{fig:condensation}.
        Note how edges correspond to possible transitions, given that you start in $V_\Gamma'$. As can be seen from Figs~\ref{fig:operating_regions_v2} and \ref{fig:prepares_graph}, transitions from $\Omega_1$ to $\Omega_2$ are possible from the set $v_a(1)$. However, $v_a(1)$ is not in $V_\Gamma'$. The only part of $\Omega_1$ in $V_\Gamma'$ is $v_b(1)$, leading to a transition to $\Omega_4$. Thus, above there is only one edge out of $\Omega_1$, and it leads to $\Omega_4$.
        }
        \label{fig:behaviour_graph}
    \end{figure}
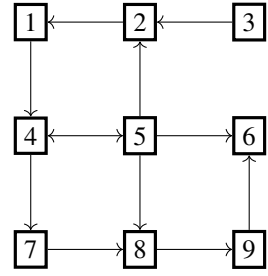

    \begin{definition}[Behavior graph]\label{def:behavior_graph}
        Let $v_P : V_\Gamma \to P$ be a surjective mapping (meaning that there exists at least one $i \in V_\Gamma$ for every $j \in P$) from the prepares graph vertices $V_\Gamma$ to the abstraction $P$, defined by
        \begin{equation}
            v_P := v_a^{-1} \cup v_b^{-1} \cup v_c^{-1},
        \end{equation}
        which is the union of the inverse mappings of (\ref{eq:va}), (\ref{eq:vb}), and (\ref{eq:vc}), which exist because they are bijective.
        The behavior graph is then defined by
        \begin{align}
            \Gamma_\Omega :=& \left(V_\Omega, E_\Omega\right) \\
            V_\Omega :=& v_p \left[V_\Gamma'\right] \\
            E_\Omega :=& \left(v_P \times v_P\right)\left[E \cap (V_\Gamma' \times V_\Gamma')\right],
        \end{align}
        where $[\cdot]$ denotes the image of a subset, 
        $(v_P \times v_P)(i,j) = (v_P(i), v_P(j))$ for $(i,j) \in E$ is the Cartesian product of functions,
        and $E \cap (V_\Gamma' \times V_\Gamma')$ is the restriction of the prepares graph's edges to the analysis set $V_\Gamma'$.
        Define the reachability relation of $\Gamma_\Omega$ as $\leq_\Omega$ (the reflexive-transitive closure of $E_\Omega$).
    \end{definition}

%  Since a single vertex in $V_\star$ of $\Gamma_\star$ can be composed of several vertices of $V_\Gamma$, the corresponding subsets of the statespace are defined as follows.

%  \begin{definition}[Sets of $V_\star$]
%      For each node  $V_\star(i) \in V_\star$ we let
%      \begin{equation}
%          v_\star(i)=\bigcup\limits_{j \in V_a(i)} v_a(j)
%          \cup
%          \bigcup\limits_{j \in V_b(i)} v_b(j)
%          \cup
%          \bigcup\limits_{j \in V_c(i)} v_c(j),
%      \end{equation}
% that is the union of the sets of the vertices from $\Gamma$ that were merged into $V_\star(i)$. 
% Similarly $v_\star'$ is the set corresponding to the vertices $V_\star'$, the sinks of $\Gamma_\star$.
%  \end{definition}
 
 We will now show that the edges in $V_\star$ do indeed capture all possible transitions between the sets making up the vertices of $V_\star$.
 
 \begin{lemma}
 \label{all_possible_transitions}
% (Need to change notation below, $S_i$ is success region, use other notation for arbitrary set!!!) 

     Given two sets $\Sigma_0$, $\Sigma_1 \in V_\star$, if an execution, (\ref{eq:bt_execution}) or (\ref{eq:bt_execution_discrete}), starts in $\bigcup \Sigma_0$ and then enters $\bigcup \Sigma_1$ without  leaving $\bigcup \Sigma_0 \cup \bigcup \Sigma_1$, then there is an edge in $E_\star$ between the vertices of the two sets, 
     $(\Sigma_0, \Sigma_1) \in E_\star$.
    %  $\exists e \in E_\star: e=(\Sigma_0, \Sigma_1)$.
 \end{lemma}
 \begin{proof}
    We will first show that the statement holds if the two sets are taken from a prepares graph $\Gamma$.
    
    If $\Sigma_0=\{v_a(i)\}$ for some $i \in P$, then by the first three lines of (\ref{eq:prepares_edges}) we see that $v_a(i)$ has edges to the vertices of all neighboring sets.
    % , so then the edge $(V_\star(i), V_\star(j))$ exists in $E_\Gamma$.
    
    If $\Sigma_0=\{v_b(i)\}$ for some $i\in P$, we know that $B_i$ is invariant while $u_i$ executes. Thus the next set needs to be a neighbor of $v_b(i)$ and also have an overlap with $B_i$. Rows 4 through 6 of (\ref{eq:prepares_edges}) covers these cases.
    % , giving an edge  $(V_\star(i), V_\star(j))$  in $E_\Gamma$.
    
    If $\Sigma_0=\{v_c(i)\}$ for some $i \in P$, we have that $\bigcup \Sigma_0$ is positively invariant since $G_i$ is positively invariant under $u_i$ and $v_c(i)=\Omega_i \cap G_i$. This violates the assumption that the state leaves $\Sigma_0$.
    
    Thus the statement holds for an ordinary prepares graph~$\Gamma$.
    For a condensed prepares graph $\Gamma_\star$, some vertices 
    % are sets of 
    correspond to multiple
    $\Gamma$ vertices, and edges exist in $\Gamma_\star$ whenever there is an edge between any elements of the two sets. Any transition between $\Sigma_0, \Sigma_1 \in V_\star$ implies a transition between some sets
    % $\Sigma_0', \Sigma_1'$ with
    $\Sigma_0' \in \Sigma_0$ and $\Sigma_1' \in \Sigma_1$,
    which is exactly the edges created in (\ref{eq:gamma_star_edges}).
    Thus the statement holds also for $\Gamma_\star$.
\end{proof}
 
 A key prerequisite of the general convergence result below is the following assumption, stating that we will leave all vertices  $V_\star$, except the sinks $\overline{V}_\star$, in finite time.
 
 \begin{assumption}[Non-infinite cycles]
 \label{ass:non_infinite_cycles}
     The execution of (\ref{eq:bt_execution}) and (\ref{eq:bt_execution_discrete}) is such that there exists an upper time bound $T$ such that for every starting state $x(0) \in \bigcup \Sigma$ for  $\Sigma \in V_\star' \setminus \overline{V}_\star$ there exists a time $T'$ smaller than the bound, $T'<T$, such that $x(T') \not \in \bigcup \Sigma$.
 \end{assumption}
 
 We are now ready to state the general convergence theorem.
 
 \begin{theorem}[General convergence]
 \label{theorem:general_convergence}
%  Under Assumption \ref{ass:fts}, 
%  let the analysis set $V_\star' \subset V_\star$ be a set of vertices such that
%  there are no edges in $E_\star$ leaving $V_L$, that is
%  \begin{equation}
%              \not \exists j \in V_\star \setminus V_L : i <_{\Gamma_\star} j,
% \end{equation}
%  and $V_\star'$ be the sinks of $V_\star$ as above.
% Furthermore, l
Let Assumptions \ref{ass:fts} and \ref{ass:non_infinite_cycles}
 hold for some $T$, 
then there is a $T'' \leq |V_\star' \setminus \overline{V_\star}|T$ such that $x(T'') \in \bigcup \bigcup \overline{V}_\star$, that is the state will reach the sinks $\overline{V}_\star$ (goal regions) in finite time upper bounded by $|V_\star' \setminus \overline{V_\star}|T$, where $|V_\star' \setminus \overline{V_\star}|$ is the cardinality (size) of $V_\star' \setminus \overline{V_\star}$.
 \end{theorem}
\begin{proof}
    We know by Lemma  \ref{all_possible_transitions} above that all possible transitions are represented as edges in $\Gamma_\star$. We also know that for vertices in $V_\star'$, all these transitions will happen within time $T$, since we know that the state will leave each set $\Sigma \in V_\star \setminus \overline{V}_\star$ corresponding to each non-sink vertex
    % $V_\star(i)$ 
    in finite time.
    The fact that there are no cycles in $\Gamma_\star$ thus implies that after time $|V_\star' \setminus \overline{V_\star}|T$ the state must have reached one of the sinks in $\overline{V}_\star$.
\end{proof}

\begin{remark}
Note that the bound above can be made less conservative by finding the longest path in $V_\star'$, which might be shorter than $|V_\star' \setminus \overline{V_\star}|$, and using this length for the upper bound. Furthermore, one could also find different time bounds for the different sets, instead of one bound $T$ for all, as in Assumption \ref{ass:non_infinite_cycles}.
\end{remark}

\begin{corollary}
\label{cor:fts}
    If Theorem  \ref{theorem:general_convergence} holds, then the whole BT $\bt_0$ is finite time successful (Definition \ref{def:fts}) with $\tau_0=|V_\star'|T$, $B_0=\bigcup \bigcup V_\star'$ and $G_0= \bigcup \bigcup (\overline{V}_\star \cap V_\star')$.
\end{corollary}
\begin{proof}
    A straightforward application of Theorem  \ref{theorem:general_convergence}.
\end{proof}

\begin{remark}
Note that Corollary \ref{cor:fts} enables the iterative application of Theorem \ref{theorem:general_convergence} on larger and larger BTs.
\end{remark}

Looking at the example of Fig. \ref{fig:condensation} we see that if we pick $V_\star'$ to be all  bold vertices (all except the bottom one) there are no edges leaving $V_\star'$. Furthermore we have $\overline{V}_\star= \{\{v_c(6)\}, \{v_c(7)\}\}$ as the sinks. The theorem now says that if the state leaves all non-sink vertices within some time that is smaller than $T$, after starting somewhere in $V_\star'$, it will reach the sinks within some time that is smaller than $|V_\star'|T=9 T$.
Similarly, if we pick $V_\star'= \{\{v_b(5)\}, \{v_b(6)\}, \{v_c(6)\}\}$, the theorem holds for this smaller set, and we start somewhere in $V_\star'$, we will reach the sinks within $|V_\star'|T=3 T$ (since we are closer to the goal).

%-------------

    \begin{figure}[ht]
        \centering
        \begin{forest}
            for tree={
                minimum size=1.5em,
                inner sep=1pt,
                l=1.2cm,
                font=\footnotesize
            }
            [$\rightarrow$, draw, tikz={\node[left=4pt of .west]  {$0$};}
                [$?$, draw, tikz={\node[left=4pt of .west]  {$1$};}
                    [{$\begin{gathered} \text{Battery} \geq 20\% \\ \land \neg \text{Charging} %\text{At Home}
                    \end{gathered}$}, draw, ellipse, tikz={\node[above left=4pt of .north west]  {$2$};}]
                    [$\rightarrow$, draw, tikz={\node[above=4pt of .north]  {$3$};}
                        [{$?$}, draw, tikz={\node[left=4pt of .west]  {$4$};}
                            [{$\begin{gathered} \text{At Home} \end{gathered}$}, draw, ellipse, tikz={\node[left=4pt of .west]  {$5$};}]
                            [$\rightarrow$, draw, tikz={\node[right=4pt of .east]  {$6$};}
                                [{$\begin{gathered} \text{Battery} > 0 \% \end{gathered}$}, draw, ellipse, tikz={\node[left=4pt of .west]  {$7$};}]
                                [{$\begin{gathered} \text{Go Home} \end{gathered}$}, draw, tikz={\node[above=4pt of .north]  {$8$};}]
                            ]
                        ]
                        [{$\begin{gathered} \text{Charge} \end{gathered}$}, draw, tikz={\node[right=4pt of .east]  {$9$};}],
                    ]
                ]
                [$?$, draw, tikz={\node[left=4pt of .west]  {$10$};}
                    [{$\begin{gathered} \text{Path} \\ \text{Surveyed} \end{gathered}$}, draw, ellipse, tikz={\node[left=4pt of .west]  {$11$};}]
                    [$\rightarrow$, draw, tikz={\node[right=4pt of .east]  {$12$};}
                        [{$?$}, draw, tikz={\node[left=4pt of .west]  {$13$};}
                            [{$\begin{gathered} \text{Near path} \end{gathered}$}, draw, ellipse, tikz={\node[above left=4pt of .north west]  {$14$};}]
                            [$\rightarrow$, draw, tikz={\node[right=4pt of .east]  {$15$};}
                                [{$\begin{gathered} \text{Battery} > 0 \% \end{gathered}$}, draw, ellipse, tikz={\node[above=4pt of .north]  {$16$};}]
                                [{$\begin{gathered} \text{Go to} \\ \text{Path} \end{gathered}$}, draw, tikz={\node[above=4pt of .north]  {$17$};}]
                            ]
                        ]
                        [{$\begin{gathered} \text{Follow Path} \end{gathered}$}, draw, tikz={\node[above=4pt of .north]  {$18$};}],
                    ]
                ]  
                [Idle, draw, tikz={\node[above=4pt of .north]  {$19$};}]
            ]
        \end{forest}
        \caption{A BC-BT for a mobile surveying robot. In this BT there exist cycles.}
        \label{fig:cycle_tree}
    \end{figure}
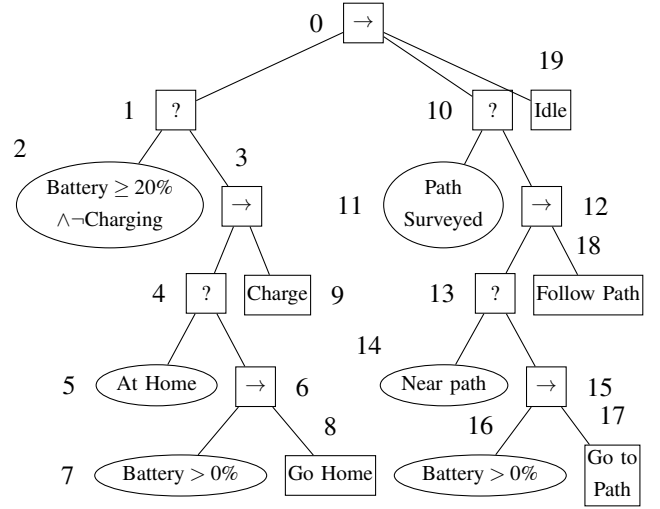
    
    \begin{figure}
        \centering
        \begin{tikzpicture}[main/.style = {draw}]
            \node[main, very thick] (8) [label=left:{$\Gamma$}] {$v_b(8)$};
            \node[main, very thick] (9) [right=1cm of 8] {$v_b(9)$};
            \node[main, very thick] (17) [right=1cm of 9] {$v_b(17)$};
            \node[main, very thick] (18) [right=1cm of 17] {$v_b(18)$};
            \node[main, very thick] (19) [below=1cm of 18] {$v_b(19)$};
            \node[main, very thick] (2) [below=1cm of 19] {$\{v_b(19)\}$};
            \node[main, very thick] (3) [left=1cm of 19] {$v_c(19)$};
            \node[main, very thick] (4) [below=1cm of 2] {$\{v_c(19)\}$};
            \node[main, very thick] (1) [label=left:{$\Gamma_\star$}, left=1cm of 2] {$\{v_b(8), v_b(9), v_b(17), v_b(18)\}$};
            \draw[->] (8) -- (9);
            \draw[->] (9) -- (17);
            \draw[->] (17) -- (18);
            \draw[->] (18) -- (19);
            \draw[->] (18) to [out=200, in=310] (8);
            \draw[->] (1) -- (2);
            \draw[->] (2) -- (4);
            \draw[->] (19) -- (3);
            % \draw[->] (10) to [out=145, in=215] (6);
        \end{tikzpicture}
        \caption{The prepares graph $\Gamma$ (top) and the condensed prepares graph $\Gamma_\star$ (bottom) corresponding to the BT in Fig. \ref{fig:cycle_tree} and Example \ref{ex:cycle_tree}. Note how \emph{Follow Path} might lead to either the battery being below 20\%, or the path being surveyed.}
        \label{fig:cycle_convergence}
    \end{figure}
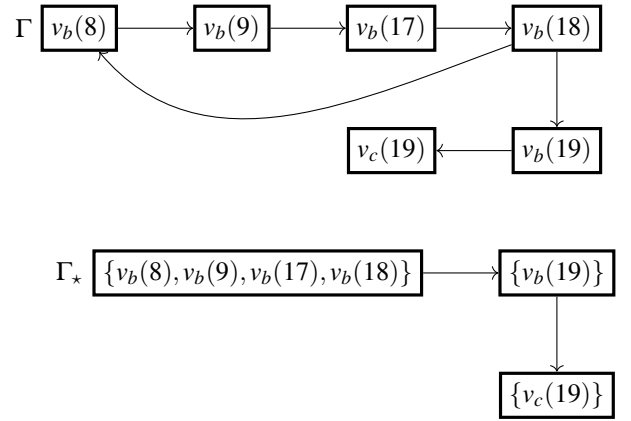
    
    \begin{example}\label{ex:cycle_tree}
        The BT in this example is following the so-called backward chained design principle, see \cite{colledanchise2019towards} for details.
        In the BT in Fig. \ref{fig:cycle_tree}, a robot is tasked with surveying a given path. However, the path is long, and the robot will need to go and recharge a number of times before completing the survey. Thus, there will be a cycle as shown in Fig. \ref{fig:cycle_convergence}.
        However, we know that eventually the path will be surveyed and, for each charging session, eventually the battery will be full.
        First, we pick the abstraction to be all the leaves of the tree that are actions (not the oval conditions that have $\Omega_i=\emptyset$), i.e.,  $P= \{8, 9, 17, 18, 19\}$.
        Then, if Assumption \ref{ass:fts} holds for this set, i.e., \emph{Go Home} will achieve \emph{At Home} in finite time and so on,
        we get the $\Gamma$ depicted in Fig. \ref{fig:cycle_convergence}. Note that in this example we also assume that $v_a(i)=\emptyset$, i.e., the actions always achieve their objectives, \emph{Go Home} will never fail to achieve \emph{At Home}.
        
        Let the set corresponding to the cycle in $\Gamma$ be $\Sigma_0 = \{v_b(8), v_b(9), v_b(17), v_b(18)\}$.
        Then, if the path is of finite length, and sufficiently close to the charging station, the robot will at some point manage to survey all of it. Thus, there exists $T \in \mathbb{R}_{>0}$ such that 
        $x(0) \in \bigcup \Sigma_0 = (\Omega_8 \cap B_8) \cup (\Omega_9 \cap B_9) \cup (\Omega_{17} \cap B_{17}) \cup (\Omega_{18} \cap B_{18})$ implies that there exists some $T'<T$ such that $x(T') \not \in \Sigma_0$.
        Theorem \ref{theorem:general_convergence} can now be applied, and states that there is a $T''\leq |V_\star'|T=3T$ such that 
        $x(T'') \in \bigcup \bigcup \overline{V}_\star=v_c(19)$, i.e. the path is surveyed, the battery is above $20 \%$ and the robot is successfully idling.
    \end{example}

    %In the theorem above, we know that there will be a finite time within which the state leaves the regions indexed by cycles.
    %However, we do not know exactly how many transitions between operating regions will take place, we only know how many operating regions can be visited (any number of times), given by $N_i$.
    %If there are no cycles in $\Gamma$, then there is no possibility visiting an operating region more than once.
    %In this case, $N_i$ serves as the maximal number of possible transitions.
    If there are no cycles in $\Gamma$, the two graphs $\Gamma$ and $\Gamma_\star$ are isomorphic and contain the same information. Furthermore, we get the following special case of the theorem above.

    \begin{corollary}[Acyclical convergence]\label{theorem:acyclical_convergence}
    If $| V_\star(i) | = 1$ for all $V_\star(i) \in V_\Gamma'$ there are no cycles in $\Gamma$.
    Then Assumption \ref{ass:non_infinite_cycles} holds by Assumption \ref{ass:fts}.
    Furthermore, if Theorem \ref{theorem:general_convergence} holds there will be no more than $|V_\star'|$ transitions between operating regions.
    \end{corollary}
    \begin{proof}
        It is clear that $| V_\star(i) | = 1$ implies the absence of cycles, as each cycle leads to an equivalence class with more than one element, i.e. $| V_\star(i) | > 1$. The absence of cycles further more implies the upper bound of $|V_\star'|$ on the number of transitions.
        Finally, Assumption \ref{ass:fts} says that all vertices in $P$ are finite time successful, with an execution that reaches some state $x' \in G_i$ in a time bounded by $\tau_i$. If this state $x' \in G_i\cup \Omega_i$ we have either transitioned, or reached a sink before $\tau_i$, and if $x' \in G_i \setminus \Omega_i$ a transition will be made before reaching $x'$, within time $\tau_i$. Either way, Assumption \ref{ass:non_infinite_cycles} is satisfied.   
    \end{proof}

    \subsection{Generalizations}\label{sec:generalisations}

    In this section, we will show how the theorems presented in Section \ref{sec:convergence} generalize previous proofs in the literature:
    \cite[Lemma 2, p.7]{colledanchise2016behavior} (Corollary \ref{cor:sequence} below),
        \cite[Lemma 3, p.8]{colledanchise2016behavior} (Corollary \ref{cor:implicit_sequence} below),
        \cite[Theorem 4, p.6]{sprague2021continuous} (Corollary \ref{cor:sprague2021continuous} below),
        \cite[Theorem 1, p.7]{ogren2020convergence} (Corollary \ref{cor:ogren2020convergence} below).

    The following two Corollaries will cover the Lemmas of Colledanchise et al. for sequences (\ref{eq:sequence}) and fallbacks (\ref{eq:fallback}) (also called \say{implicit sequences} in the context of convergence).
    \begin{corollary}[{\cite[Lemma 2, p.7]{colledanchise2016behavior}}]\label{cor:sequence}
        If $\bt_1, \bt_2$ are finite-time successful,
        $\bt_0 = Seq[\bt_1, \bt_2]$ and 
        % $\partial (\Omega_1 \cap B_1) \cap \partial (\Omega_2 \cap B_2) \neq \emptyset$ 
        $S_1 = B_2$,
        % $B_1 \setminus R_1 \subseteq B_2$
        % and $B_2 \subseteq S_1$, 
        then $\bt_0$ is finite-time successful with $B_0 = B_1 \cup B_2$.
    \end{corollary}
    \begin{proof}
        According to (\ref{eq:operating}), the leaf operating regions are given by
        $\Omega_1 = R_1 \cup F_1$ and
        $\Omega_2 = S_1 \cap (R_2 \cup S_2 \cup F_2) = S_1$,
        but since $S_1 = B_2$, we must have that 
        the second sub-BT's operating region is fully composed of its DOA, i.e. $\Omega_2 = B_2$. 
        Thus, according to (\ref{eq:prepares_edges}), we must have 
        % $(v_a(2), v_b(2)) \not \in E$,
        $(v_b(2), v_b(1)) \not \in E_\Gamma$
        and
        $(v_b(2), v_c(2)) \in E_\Gamma$.
        Additionally, since the first sub-BT's goal region is in the success region of its DOA, i.e. $G_1 \subseteq B_1 \cap S_1$ (which is contained in $\Omega_2$), we must have
        $(v_b(1), v_c(1)) \not \in E_\Gamma$ 
        and 
        $(v_b(1), v_b(2)) \in E_\Gamma$.

        % Since $B_1 \subseteq R_1 \cup S_1$,
        % we have that $B_1 \setminus \Omega_1 = B_1 \setminus R_1$.

        % Since $\Omega_2 = S_1$, $B_2 \subseteq S_1$, and $B_1 \setminus R_1 \subseteq B_2$, we must have that $B_1 \setminus \Omega_1 \subseteq \Omega_2 \cap B_1$.
        % Thus, we must have $(v_b(1) , v_b(2)) \in E$.
        % Since $G_1 \subseteq B_1 \cap S_1$ and $\Omega_1 = S_1$, we have that $\Omega_1 \cap G_1 = \emptyset$. Thus $(v_b(1), v_c(1)) \not \in E$.
        % Since $B_2 \subseteq S_1$ and $\Omega_2 = S_1$, we have that $B_2 \subseteq \Omega_2$,
        % thus $(v_b(2), v_b(1)) \not \in E$ and $(v_b(2), v_c(2)) \in E$.
        
        Thus, $V_\Gamma' = \{v_b(1), v_b(2), v_c(2)\}$ is a valid analysis set. % with a corresponding  analysis set $L_\star = \{\{v_b(1)\}, \{v_b(2)\}, \{v_c(2)\}\}$.
        Furthermore, if Corollary \ref{theorem:acyclical_convergence} holds, then by Corollary \ref{cor:fts}, Assumption \ref{ass:fts} holds for the whole BT $\bt_0$ with $B_0 = B_1 \cup B_2$ and $\tau_0 = \tau_0 + \tau_1$.
    \end{proof}

    \begin{corollary}[{\cite[Lemma 3, p.8]{colledanchise2016behavior}}]\label{cor:implicit_sequence}
        If $\bt_1, \bt_2$ are finite-time successful,
        $\bt_0 = Fal[\bt_1, \bt_2]$,
        % $B_2 \setminus \Omega_2 \subseteq \Omega_1 \cap B_1$,
        $B_1 = R_1 \cup S_1$,
        and 
        $S_2 \subseteq B_1$
        % and $B_2 \subseteq \Omega_1$, 
        then $\bt_0$ is finite-time successful with $B_0 = B_1 \cup B_2$.
    \end{corollary}
    \begin{proof}
        According to (\ref{eq:operating}), the leaf operating regions are given by
        $\Omega_1 = R_1 \cup S_1$ and
        $\Omega_2 = F_1 \cap (R_2 \cup S_2 \cup F_2) = F_1$,
        but $S_2 \subseteq B_1$,
        thus $\Omega_2 = F_1 \cap (R_2 \cup F_2)$.
        Also, we have that $\Omega_1 = B_1$,
        thus $(v_b(1), v_b(2)) \not \in E_\Gamma$
        and $(v_b(1), v_c(1)) \in E_\Gamma$.
        Since the goal region of the second sub-BT is contained in the success part of its DOA, i.e. $G_2 \subseteq B_2 \cap S_2$ (which is contained in $\Omega_1$),
        we must have $(v_b(2), v_c(2)) \not \in E_\Gamma$
        and $(v_b(2), v_b(1)) \in E_\Gamma$.
        
        % Thus, we have that $B_2 \setminus F_1 = B_2 \setminus \Omega_2$ and $B_1 = \Omega_1 \cap B_1$.
        % Thus, $B_2 \setminus F_1 \subseteq \Omega_1 \cap B_1$,
        % hence $(v_b(2), b_v(1)) \in E$.
        % Furthermore, since $S_2 \subseteq R_1 \cap S_1$ and $G_2 \subseteq B_2 \cap S_2$, we have that $(v_b(2), v_c(2)) \not \in E$.
        % Furthermore, we have that $B_1 \subseteq \Omega_1$,
        % thus $(v_b(1), v_b(2)) \not \in E$.

        Thus, $V_\Gamma' = \{v_b(2), v_b(1), v_c(1)\}$ is a valid analysis set with a corresponding condensed analysis set $V_\star' = \{\{v_b(1)\}, \{v_b(2)\}, \{v_c(2)\}\}$.
        Thus, if Corollary \ref{theorem:acyclical_convergence} holds, then by Corollary \ref{cor:fts},  Assumption \ref{ass:fts} holds for the whole BT $\bt_0$ with $B_0 = B_1 \cup B_2$ and $\tau_0 = \tau_0 + \tau_1$.
    \end{proof}

    For the results of Sprague et al. \cite[Theorem 4, p.6]{sprague2021continuous} to be covered by Theorem \ref{theorem:general_convergence}, we need to add  Assumption \ref{ass:fts}, since this aspect was dealt with slightly differently in \cite{sprague2021continuous}.

     \begin{corollary}[{\cite[Theorem 4, p.6]{sprague2021continuous}}]\label{cor:sprague2021continuous}
     If Assumption \ref{ass:fts} holds, there exists a subset $L \subseteq P$ and a partial order $\leq_f \subset L^2$
    such that the region
    
    \begin{equation}\label{eq:constraint}
        % \Lambda_i := \bigcup_{j \geq_f i} \Omega_j \setminus F_j
        \Lambda_i := \bigcup_{j \geq_f i} \Omega_j \setminus F_0
    \end{equation}
    is invariant under $f(x, u_i(x))$ for all $i \in L$,
    and there exists a finite time $\tau_i > 0$,
    such that if 
    $x(t) \in \Omega_i \setminus S_0$
    then 
    $x(t + \tau_i) \not \in \Omega_i \setminus S_0$
    for all $i \in L$,
    then there exists a maximum number of transitions
    $N \in \mathbb{N}$ and a maximum duration $t' > 0$,
    such that if $x(0) \in \Lambda_i$ for any $i \in L$,
    then 
    $x(t) \in S_0$ 
    in bounded time $t \leq t'$ within $N$ transitions.
    \end{corollary}
    \begin{proof}
    The abstraction $P$ is the same in the two papers (obeying Definition \ref{def:abstraction}), and $L$ corresponds to the analysis set $V_\Gamma'$.
    We will first show that the statements above imply that $\Gamma$ has no cycles.
    
    The definition of $\Lambda_i$ above implies that 
    $\Lambda_i \subseteq \Lambda_j$ if $i \geq_f j$. The fact that all $\Lambda_i$ are invariant implies that if $i \geq_f j$ there can only be a transition from $\Omega_j$ to $\Omega_i$ and not the other way around. Therefore $\Gamma$ cannot have any cycles.
    
    Furthermore, $S_0$ is the global success region, corresponding to the sinks $\overline{V}_\star$.
    The fact that the state always leaves $\Omega_i \setminus S_0$ in finite time implies that Assumption~\ref{ass:non_infinite_cycles} holds.

    Now, Theorem \ref{theorem:general_convergence} and Corollary \ref{cor:fts} imply that there are no more than 
    $|V_\star'|$ transitions and $x(T'') \in \bigcup \bigcup \overline{V}_\star$ for some bounded time $T''$.
    Furthermore, it is clear that $\leq_\Omega$ (Definition \ref{def:behavior_graph}) would be a partial order and isomorphic to $\leq_f$.
    \end{proof}

    \subsection{Backchaining}
    
    In this section, we will show how the presented theorem generalizes the theorem on so-called backchained BTs \cite[Theorem 1, p.7]{ogren2020convergence}.
    A \textit{backchained BT} (BC-BT) \cite{colledanchise2019towards, ogren2020convergence} is a BT  that is automatically generated from a set of \textit{action} and \textit{condition} BTs, where each action BT has a
    set of \textit{preconditions} that need to be satisfied in order for the action to work, see Table \ref{tab:action}, and each condition has a set of actions that achieves it, see Table \ref{tab:condition}.
    With this, we formalize the concept of \textit{action and condition libraries}, as follows.

    \begin{definition}[Action and condition libraries]\label{def:libraries}
        \label{def:post-pre-conditions}
        Mutually disjoint finite index sets, $A$ and $C$, respectively, that satisfy the following.
        
        Action library $A$: For all $i \in A$, there exists a FTS BT $\bt_i$ (Def. \ref{ass:fts}) and 
        a minimal totally ordered subset $(C_i, \leq_{C_i})$,
        such that $\bigcap_{j \in C_i} S_j = B_i$ and $\{C_i\}_{i \in A}$ is pairwise disjoint.
        The subset $C_i \subseteq C$ is the preconditions of action $i \in A$; it is minimal because we do not want to have to satisfy extra unneeded conditions.
        
        Condition library $C$: For all $i \in C$, there exists a BT $\bt_i$
        with
        $R_i = \emptyset$ (a condition) 
        and a maximal totally ordered subset $(A_i, \leq_{A_i})$,
        such that,
        $\bigcup_{j \in A_i} S_j \subseteq S_i$, and $\{A_i\}_{i \in C}$ is pairwise disjoint.
        The subset $A_i \subseteq A$ is the actions whose postcondition is $i \in C$; it is maximal because we want as many actions as possible to fallback to in order to achieve the postcondition.
        
        % The subset $C_i \subseteq C$ is the preconditions of action $i \in A$; it is minimal because we do not want to have to satisfy extra unneeded conditions.
        
        % The subset $A_i \subseteq A$ is the actions whose postcondition is $i \in C$; it is maximal because we want as many actions as possible to fallback to in order to achieve the postcondition.
        
        % such that there exists a FTS BT $\bt_i$ for all $i \in A$,
        % $R_i = \emptyset$ for all $i \in C$,
        % for all $i \in C$ there exists a maximal totally ordered subset $(A_i, \leq_{A_i})$ of $A$
        % such that $\bigcup_{j \in A_i} S_j \subseteq S_i$,
        % and for all $i \in A$ there exists a minimal totally ordered subset $(C_i, \leq_{C_i})$ of $C$
        % such that $\bigcap_{j \in C_i} S_j \subseteq B_i$.
        % The subset $C_i \subseteq C$ is the preconditions of action $i \in A$; it is minimal because we do not want to have to satisfy extra unneeded conditions.
        % The subset $A_i \subseteq A$ is the actions whose postcondition is $i \in C$; it is maximal because we want as many actions as possible to fallback to in order to achieve the postcondition.
    \end{definition}

    \begin{table}[!h]
        \centering
                \caption{An example action library}
        \begin{tabular}{l|l}
           Actions $i \in A$ & Pre-condition(s) $C_i$ \\
           \hline
           Grasp  & Object reachable \\
           Place at goal  & Object in gripper, Near goal \\
           Goto safe area & Safe area reachable \\
           Idle & In safe area, Object at goal \\
           \vdots  & \vdots 
        \end{tabular}
        \label{tab:action}
    \end{table}
    
    \begin{table}[!h]
        \centering
                \caption{An example condition library}
        \begin{tabular}{l|l}
           Conditions $i \in C$ & Action(s) $A_i$ that achieves $i$  \\
           \hline
           Object at goal  & Place at goal \\
           In safe area & Goto safe area \\
           Near object  & Goto object \\
           \vdots  & \vdots 
        \end{tabular}
        \label{tab:condition}
    \end{table}
    
     Note that the $C_i$ are indeed preconditions in the sense that we require $\bigcap_{j \in C_i} S_j = B_i$ i.e. that the intersection of the success regions $S_i$ of the conditions is the DOA of the action, and $R_i=\emptyset$ i.e., they are conditions (that never return running), not actions. 
    
    Similarly, $A_i$ do indeed achieve condition $i$ in the sense that $\bigcup_{j \in A_i} S_j \subseteq S_i$, i.e., the success region of the actions are inside the success region of the condition.
    
    Now, given an action library and a condition library, we can automatically construct a BC-BT, as shown in the following definition and Fig. \ref{fig:bc-bt}, following \cite[Algorithm 1, p.4]{ogren2020convergence}
    An example of a BC-BT is illustrated in Fig. \ref{fig:bc-bt}.
    
    \begin{definition}[Backchained BT]\label{def:bcbt}
        A BT defined recursively for an action $i \in A$ as
        \begin{equation}\label{eq:bcbt}
            \Pi\left[i\right] := 
                Seq \left[ 
                    \left(
                        Fal \left[
                            \bt_j, 
                            \left(
                                \Pi[k]
                            \right)_{k \in A_j}
                        \right]
                    \right)_{j \in C_i},
                    \bt_i
                \right]
        \end{equation}
        where
        $(\cdot)_{j \in A_i}$
        and
        $(\cdot)_{k \in C_j}$
        are enumerations of
        $A_i$
        and $C_j$
        that are totally ordered $\leq_{A_i}$ and $\leq_{C_j}$, respectively.
        % and
        % $[\bt_i, (\cdot)_{k \in A_i}]$
        % and
        % $[(\cdot)_{j \in C_i}, \bt_j]$
        % are concatenations.
        Note that the recursion (\ref{eq:bcbt}) will halt when either
        $C_i = \emptyset$ or $A_j = \emptyset$.
    \end{definition}

    \begin{figure}[ht]
        \centering
        \begin{forest}
            for tree={
                minimum size=1.5em,
                inner sep=1pt,
                l=1.2cm,
                font=\footnotesize
            }
            [$\rightarrow$, draw
                [$?$, draw
                    [{$\begin{gathered}\text{In safe} \\ \text{area} \end{gathered}$}, draw, ellipse, fill=SeaGreen,  tikz={\node[below=1pt of .south]  {$1$};}]
                    [$\rightarrow$, draw
                        [$?$, draw
                            [{$\begin{gathered}\text{Safe area} \\ \text{reachable}\end{gathered}$}, draw, ellipse, tikz={\node[below=1pt of .south]  {$2$};}]
                        ]
                        [{$\begin{gathered}\text{Go to} \\ \text{safe area}\end{gathered}$}, draw, tikz={\node[below=1pt of .south] {$9$};}]
                    ]
                ]
                [$?$, draw
                    [{$\begin{gathered}\text{Object at} \\ \text{goal}\end{gathered}$}, draw, ellipse, fill=BurntOrange, tikz={\node[above=1pt of .north]  {$3$};}]
                    [$\rightarrow$, draw
                        [$?$, draw
                            [{$\begin{gathered}\text{Object in} \\ \text{gripper}\end{gathered}$}, draw, ellipse, fill=SeaGreen, tikz={\node[below=1pt of .south]  {$4$};}]
                        [{$\rightarrow$}, draw
                            [$?$, draw,
                                [{$\begin{gathered}\text{Near} \\ \text{object}\end{gathered}$}, draw, ellipse, tikz={\node[left=1pt of .west]  {$5$};}]
                                [$\rightarrow$, draw
                                    [$?$, draw
                                    [{$\begin{gathered} \text{Object} \\ \text{reachable}\end{gathered}$}, draw, ellipse, tikz={\node[left=1pt of .west]  {$6$};}]
                                    ]
                                    [{$\begin{gathered} \text{Go to} \\ \text{object}\end{gathered}$}, draw, tikz={\node[below=1pt of .south]  {$10$};}]
                                ]
                            ]
                            [{$\begin{gathered} \text{Grasp} \\ \text{object}\end{gathered}$}, draw, tikz={\node[below=1pt of .south] {$11$};}]
                        ]
                        ]
                        [$?$, draw
                            [{$\begin{gathered}\text{Near} \\ \text{goal} \end{gathered}$}, draw, ellipse, fill=BurntOrange, tikz={\node[above=1pt of .north]  {$7$};}]
                            [$\rightarrow$, draw
                                [$?$, draw 
                                [{$\begin{gathered}\text{Goal} \\ \text{reachable} \end{gathered}$}, draw, fill=Cerulean, ellipse, tikz={\node[below=1pt of .south]  {$8$};}]
                                ]
                                [{$\begin{gathered}\text{Go to} \\ \text{goal} \end{gathered}$}, draw, very thick, tikz={\node[below=1pt of .south]  {$12$};}]
                            ]
                        ]
                        [{$\begin{gathered} \text{Place} \\ \text{object}\end{gathered}$}, draw, tikz={\node[below=1pt of .south]  {$13$};}]
                    ]
                ]
                [{$\begin{gathered} \text{Idle} \\ \text{(all objectives achieved)}\end{gathered}$}, draw, tikz={\node[above=1pt of .north]  {$14$};}]
                %[{$\begin{gathered} \text{Protect} \\ \text{object}\end{gathered}$}, draw, tikz={\node[above=1pt of .north]  {$14$};}]
            ]
        \end{forest}
        \caption{A BC-BT for adapted from the mobile manipulator example in \cite{ogren2020convergence}.
        The success regions of $\bt_j : j \in C_{12}^\text{ACC}$ (green) and $\bt_j : j \in C_{12}$ (blue) are positively invariant with respect to the execution of $\bt_{12}$ (bolded) if $B_i \subseteq \bigcap_{j \in C_{12}^\text{ACC}} S_j$.
        However, the failure regions of $\bt_j : j \in C_{12}^\text{Post}$ (orange) are not positively invariant (see Corollary \ref{cor:ogren2020convergence}).
        }
        \label{fig:bc-bt}
    \end{figure}
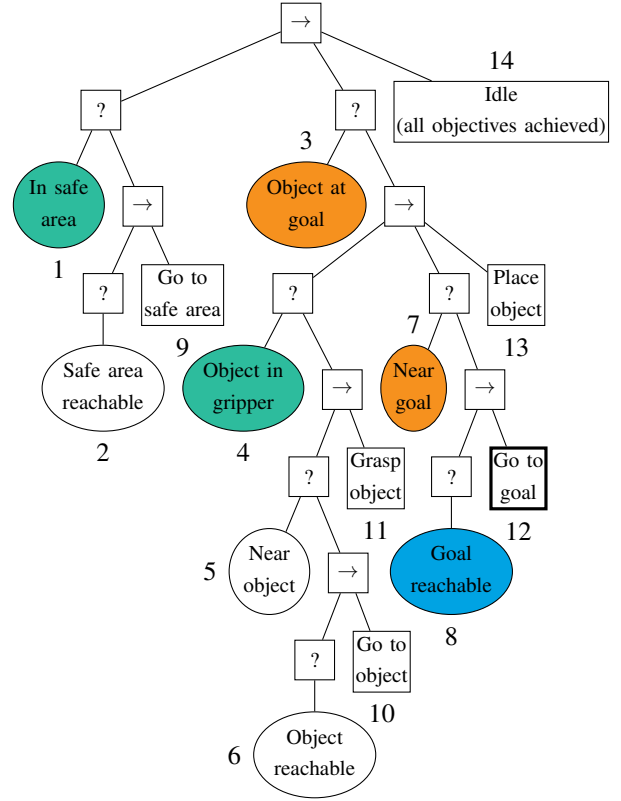

\begin{example}
    To illustrate how the BT in Fig.~\ref{fig:bc-bt} was created we imagine starting with the action Idle, in the last row of Table~\ref{tab:action}. Here the pre-conditions are not given by what is physically possible, but rather by the will of the user. The user wants the robot to idle only when it is in a safe area and the object is at the goal. Now imagine that Table~\ref{tab:condition} was empty. Then the recursion of Equation (\ref{eq:bcbt}) would end immediately, with $A_j=\emptyset$ and we would get the BT in Fig.~\ref{fig:top_level_bc-bt}, which is similar to the top level in Fig.~\ref{fig:bc-bt}.
    If Tables ~\ref{tab:action},\ref{tab:condition} were as given above, the recursion would continue a few more steps and we would get the upper half of Fig.~\ref{fig:bc-bt}. Similarly, one can imagine that with a few extra lines in each table, we would get the full BT of Fig.~\ref{fig:bc-bt}.
\end{example}

 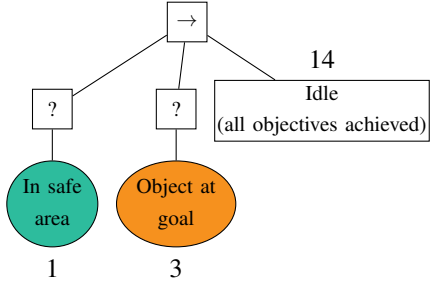
\begin{figure}[ht]
        \centering
        \begin{forest}
            for tree={
                minimum size=1.5em,
                inner sep=1pt,
                l=1.2cm,
                font=\footnotesize
            }
            [$\rightarrow$, draw
                [$?$, draw
                    [{$\begin{gathered}\text{In safe} \\ \text{area} \end{gathered}$}, draw, ellipse, fill=SeaGreen,  tikz={\node[below=1pt of .south]  {$1$};}]
                ]
                [$?$, draw
                    [{$\begin{gathered}\text{Object at} \\ \text{goal}\end{gathered}$}, draw, ellipse, fill=BurntOrange, tikz={\node[below=1pt of .south]  {$3$};}]
                ]
                [{$\begin{gathered} \text{Idle} \\ \text{(all objectives achieved)}\end{gathered}$}, draw, tikz={\node[above=1pt of .north]  {$14$};}]
            ]
        \end{forest}
        \caption{The BT resulting from applying Equation (\ref{eq:bcbt}) to the action Idle of Table~\ref{tab:action}, if Table~\ref{tab:condition} was empty (resulting in no recursive calls in (\ref{eq:bcbt})). Note the similarity to the top part of Fig.~\ref{fig:bc-bt}. If the top-level goals of being in the safe area and having the object at the goal are satisfied, the robot can idle. However, this small BT has no actions for achieving those goals. To add those, Table~\ref{tab:condition} needs to be non-empty, leading to more steps in the recursion being executed.
        }
        \label{fig:top_level_bc-bt}
    \end{figure}

    The metadata regions of the sub-BTs in a BC-BT (\ref{eq:bcbt}) can be computed similarly to (\ref{eq:sequence-metadata}) and (\ref{eq:fallback-metadata}).
    But, instead of analyzing sequences and fallbacks as before, here we will analyze \textit{action sub-BTs} and \textit{condition sub-BTs}.

    An action sub-BT ($\bt_{p(i)} : i \in A$) in a BC-BT will either be a Sequence with an action as its only child
    or a Sequence with an action and its set of preconditions ($j \in C_i$).
    Each of such preconditions will be a condition sub-BT ($\bt_{p(j)} : j \in C$), which will either be a Fallback with a condition as its only child or
    a Fallback with a condition and its set of actions ($k \in A_j$) with the condition as the actions' postcondition.
    The metadata regions of these sub-BTs can be computed explicitly by \textit{mutual recursion} in terms of action and condition libraries (Def. \ref{def:libraries}), as shown  in the following lemma.

    \begin{lemma}
        An action sub-BT's ($\bt_{p(i)} : i \in A$) metadata regions are
        \begin{equation}\label{eq:action-metadata}
            \begin{aligned}
                R_{p(i)} &= \left(R_i \bigcap_{j \in C_i} S_j\right) \cup \left(\bigcup_{j \in C_i}\left(R_{p(j)} \bigcap_{k <_{C_i} j} S_k\right)\right) \\
                S_{p(i)} &= S_i \bigcap_{j \in C_i} S_j \\
                F_{p(i)} &= \bigcup_{j \in C_i}\left(F_{p(j)} \bigcap_{k <_{C_i} j} S_k\right),
            \end{aligned}
        \end{equation}
        where $R_{p(j)},F_{p(j)}$ refers to condition sub-BTs. 
        A condition sub-BT's ($\bt_{p(i)} : i \in C$) metadata regions are
        \begin{equation}\label{eq:condition-metadata}
            \begin{aligned}
                R_{p(i)} &= \bigcup_{j \in A_i}\left(R_{p(j)} \bigcap_{k <_{A_i} j} \left(F_i \cap F_{p(k)}\right)\right) \\
                S_{p(i)} &= S_i  \\
                F_{p(i)} &= F_i \bigcap_{j \in A_i} F_{p(j)},
            \end{aligned}
        \end{equation}
        where $R_{p(j)},F_{p(j)}, F_{p(k)}$ refers to action sub-BTs.
        A mutual recursion is formed by (\ref{eq:action-metadata}) and (\ref{eq:condition-metadata}).
    \end{lemma}
    \begin{proof}
        For an action ($i \in A$) we have that $\bt_{p(i)}$ (action sub-BT) is a Sequence with $\{p(j) : j \in C_i\} \cup \{i\}$ as its set of children,
        where $\bt_{p(j)}$ are condition sub-BTs.
        The metadata regions of a Sequence is given in terms of its children's' metadata regions in (\ref{eq:sequence-metadata}).
        % \cite[Lemma 1, p.4]{sprague2021continuous}.
        This gives 
        the running, success, and failure regions for an action sub-BT as
        \begin{equation}
            \begin{aligned}
                R_{p(i)} &= \left(R_i \bigcap_{j \in C_i} S_{p(j)}\right) \cup \left(\bigcup_{j \in C_i}\left(R_{p(j)} \bigcap_{k <_{C_i} j} S_{p(k)}\right)\right) \\
                S_{p(i)} &= S_i \bigcap_{j \in C_i} S_{p(j)} \\
                F_{p(i)} &= \left(F_i \bigcap_{j \in C_i} S_{p(j)}\right) \cup \left(\bigcup_{j \in C_i}\left(F_{p(j)} \bigcap_{k <_{C_i} j} S_{p(k)}\right)\right).
            \end{aligned}
        \end{equation}
        respectively.

        For a condition ($i \in C$) we have that $\bt_{p(i)}$ (condition sub-BT) is a Fallback with $\{i\} \cup \{p(j) : j \in A_i\}$ as its set of children,
        where $\bt_{p(j)}$ are action sub-BTs.
        The metadata regions of a Fallback are given in terms of its children's' metadata regions in (\ref{eq:fallback-metadata}).
        % \cite[Lemma 2, p.4]{sprague2021continuous}.
        This gives
        the running, success, and failure regions for a condition sub-BT as
        \begin{equation}
            \begin{aligned}
                R_{p(i)} &= R_i \bigcup_{j \in A_i}\left(R_{p(j)} \bigcap_{k <_{A_i} j} \left(F_i \cap F_{p(k)}\right)\right) \\
                S_{p(i)} &= S_i \bigcup_{j \in A_i}\left(S_{p(j)} \bigcap_{k <_{A_i} j} \left(F_i \cap F_{p(k)}\right)\right) \\
                F_{p(i)} &= F_i \bigcap_{j \in A_i} F_{p(j)},
            \end{aligned}
        \end{equation}
        respectively.

        By Def. \ref{def:libraries}, we have the following.
        For all $i \in C$, we have that $R_i = \emptyset$ implies 
        $R_{p(i)} = \bigcup_{j \in A_i}\left(R_{p(j)} \bigcap_{k <_{A_i} j} F_{p(k)}\right) $ and
        $\bigcup_{j \in A_i} S_j \subseteq S_i$ implies $S_{p(i)} = S_i$.
        For all $i \in A$, we have that $\bigcap_{j\in C_i} S_j = B_i$ and $B_i \subseteq R_i \cup S_i$ imply $F_{p(i)} = \bigcup_{j \in C_i}\left(F_{p(j)} \bigcap_{k <_{C_i} j} S_{p(k)}\right)$.
        Thus, we have (\ref{eq:action-metadata}) and (\ref{eq:condition-metadata}).
    \end{proof}
    
    The BC-BT defined in (\ref{eq:bcbt}) is such that the success region of its actions is within the success region of their postconditions, i.e. $\bigcup_{j \in A_i} S_j \subseteq S_i$.
    Similarly, the success regions of
    preconditions for a particular action equals its DOA, i.e. $\bigcap_{j \in C_i} S_j = B_i$.
    Thus, the expected behavior in a BC-BT is that actions will execute one by one, where the satisfaction of one action's postcondition also satisfies one of another action's preconditions. 
    If the postcondition of action $i$ satisfies a precondition of an action $k$ we say that they are \say{linked}. The pair of actions, together with the condition are called a \say{link}. Examples of such links include (Goto object, Near object, Grasp object) and (Grasp object, Object in gripper, Place object).

    In detail, if there are two actions $(i,k) \in A^2$ such that there exists a condition $j \in C$ such that $i \in A_j$ ($\bt_j$ is the postcondition of $\bt_i$) and $j \in C_k$ ($\bt_j$ is a precondition of $\bt_k$), then we say that $i$ and $k$ are linked and that $(i,j,k) \in \lambda$ is a \say{link}, where $\lambda \subseteq A \times C \times A$ is the set of all such links.
    
    These links induce a partial order of actions $\leq_\lambda \subseteq A^2$ called a \say{link order}. These concepts are formalized in the following definition.
    
    % In \cite{ogren2020convergence}, BC-BTs have one action for each postcondition.

    % \begin{assumption}
    %     For all $i \in C$ assume $|A_i| = 1$.
    % \end{assumption}
    
    % \begin{lemma}
    %     For an action $i \in A$,
    %     the postconditions that need to return failure for the action to execute are
    %     \begin{equation}
    %         C_i^F = \bigcup_{(p,q,r) \in \lambda_i} \left\{q\right\},
    %     \end{equation}
    %     the conditions aside from the postconditions that need to return success for the action to execute (i.e. the \say{active constraint conditions} of \cite[Def. 7, P.5]{ogren2020convergence}) are
    %     \begin{equation}
    %         C_i^S = \bigcup_{(p,q,r) \in \lambda_i} \left\{j \mid j <_{C_r} q\right\},
    %     \end{equation}
    %     the actions whose sub-BTs need to return failure for the action to execute are
    %     \begin{equation}
    %         A_i^F = \bigcup_{(p,q,r) \in \lambda_i} \left\{j \mid j <_{A_q} p\right\}.
    %     \end{equation}
    % \end{lemma}

    \begin{definition}[Links]\label{def:links}
       The set of all links is 
       \begin{equation}
            \lambda := \left\{ \left(a,b,c\right) \in A \times C \times A \mid \left(a \in A_b\right) \land \left(b \in C_c \right) \right\}
        \end{equation}
         and
        \begin{equation}
            \leq_\lambda := \left\{\left(a,c\right) \in A^2 \mid \exists b \in C : \left(a,b,c \right) \in \lambda\right\}^{=,+}
        \end{equation}
        is a partial order of actions,
        where ${}^{=,+}$ is the reflexive-transitive closure operation.
        Furthermore, the set of action-condition-action triplets defined for an action ($i \in A$) is
        \begin{equation}
            \lambda_i := \left\{\left(a,b,c\right) \in \lambda \middle| i \leq_\lambda a \right\}.
        \end{equation}
    \end{definition}
    Thus we have that e.g.
    $\lambda_{\mbox{\footnotesize Goto object}}$ includes 
    (Grasp object, Object in gripper, Place object).

    In the remainder of this section, we will use the ideas above to compute the operating regions within a BC-BT.
    Before doing so, we will make the following assumption to simplify the analysis and align more closely to \cite[Theorem 1, p.7]{ogren2020convergence}.

    % In doing so, we will show how Corollary \ref{cor:sprague2021continuous} generalises \cite[Theorem 1, p.7]{ogren2020convergence}.
    % B

    \begin{assumption}\label{as:bcbt}
        Given a BC-BT $\Pi[i]$ for $i \in A$,
        assume $\lambda_i = \emptyset$.
        In other words,
        there are no postconditions for the top-level action to satisfy, and there are no other actions that have such postconditions as their preconditions.
        Additionally, for all $i \in C$, assume $|A_i| = 1$ (there is one action per postcondition), and
        for all $i \in C$ such that $A_i = \emptyset$, assume $F_i = \emptyset$ (conditions that are not the postcondition for any action cannot fail).
    \end{assumption}
    Looking at Fig. \ref{fig:bc-bt} the above assumption is satisfied as Idle does not satisfy any postconditions and conditions such as Safe area reachable and Object reachable never fail (if they do, there is nothing that can be done about it).
    
    We will now give a more detailed example, illustrating the concepts described so far in relation to Fig. \ref{fig:bc-bt}.
    
    \begin{example}
    The set of actions and conditions in the BC-BT $\Pi[14]$ (shown in Fig. \ref{fig:bc-bt}) are given by
    \begin{equation}
        \begin{aligned}
            C = \left\{1,2,3,4,5,6,7,8\right\}, \quad
            A = \left\{9,10,11,12,13,14\right\},
        \end{aligned}
    \end{equation}
    respectively.
    The preconditions of actions are given by
    \begin{equation}
        \begin{gathered}
            C_{9} = \left\{2\right\}, \quad
            C_{10} = \left\{6\right\}, \quad
            C_{11} = \left\{5\right\} \\
            C_{12} = \left\{8\right\}, \quad
            C_{13} = \left\{4, 7\right\}, \quad
            C_{14} = \left\{1, 3\right\}
        \end{gathered} 
    \end{equation}
    The actions for postconditions is given by
    \begin{equation}
        \begin{gathered}
            A_{1} = \left\{9\right\}, \quad
            A_{2} = \emptyset, \quad
            A_{3} = \left\{13\right\}, \quad
            A_{4} = \left\{11\right\}, \\
            A_{5} = \left\{10\right\}, \quad
            A_{6} = \emptyset, \quad 
            A_{7} = \left\{12\right\}, \quad
            A_{8} = \emptyset.
        \end{gathered}
    \end{equation}
    The links are given by
    \begin{equation}
    \begin{aligned}
        \lambda =
            \{
                &\left(9, 1, 14\right),
                \left(10, 5, 11\right),
                \left(11, 4, 13\right),\\
                &\left(12, 7, 13\right),
                \left(13, 3, 14\right)
            \}.
    \end{aligned}
    \end{equation}
    The link order is given by
    \begin{equation}
        \begin{aligned}
            \leq_\lambda =
                \{
                    &\left(9,9\right),
                    \left(9,14\right), \\
                    &\left(10,10\right),
                    \left(10,11\right),
                    \left(10,13\right),
                    \left(10,14\right), \\
                    &\left(11,11\right),
                    \left(11,13\right),
                    \left(11,14\right), \\
                    &\left(12,12\right),
                    \left(12,13\right),
                    \left(12,14\right), \\
                    &\left(13,13\right),
                    \left(13,14\right),
                    \left(14,14\right)
            \}.
        \end{aligned}
    \end{equation}
    For $i = 12$ (bolded in Fig. \ref{fig:bc-bt}), we have
    \begin{equation}
        \lambda_{12} = \left\{
            \left(12, 7, 13\right),
            \left(13, 3, 14\right)
        \right\}.
    \end{equation}
    \end{example}

    In the following lemma, we will show that 
    if an action executes, it will not make the entire BT return Failure, $\Omega_i \cap F_i = \emptyset$ and it will not make the entire BT return success if it has a link, $\lambda_i \neq \emptyset \implies \Omega_i \cap S_i = \emptyset$. The latter implies that in the example of Fig. \ref{fig:bc-bt}, the BT returns success only when Idle (all objectives achieved) returns success. Finally, the conditions never return running, as expected, i.e. for $i \in C$, we have that $\Omega_i = \emptyset$.

    \begin{lemma}[Backchained operating regions]\label{lemma:backchained_operating}
        For all $i \in A$ we have that $\Omega_i \cap F_i = \emptyset$, and for all $i \in A$ such that $\lambda_i \neq \emptyset$ we have that $\Omega_i \cap S_i = \emptyset$.
        For all $i \in C$, we have that $\Omega_i = \emptyset$.
    \end{lemma}
    \begin{proof}
        First, observe that the only node that will be in the success pathway $\mathfrak{S}$ is the action $i \in A$ that defines the BC-BT $\Pi[i]$ (where $\lambda_i = \emptyset$ by Assumption \ref{as:bcbt}).
        This is because all other nodes in the BC-BT have right uncles whose parents are sequences, see (\ref{eq:success-pathway}) and Fig. \ref{fig:bc-bt}.
        Second, observe that all actions $i \in A$ will be in the failure pathway $i \in \mathfrak{F}$.
        This is because all actions do not have any right uncles whose parents are fallbacks, see (\ref{eq:failure-pathway}) and Fig. \ref{fig:bc-bt}.
        Thirdly, observe that all conditions $i \in C$ such that they are not a postcondition of any actions, $A_i = \emptyset$, will be in the failure pathway $i \in \mathfrak{F}$.
        This is because all conditions that are not a postcondition for any action do not have any right uncles whose parents are fallbacks, see (\ref{eq:failure-pathway}) and Fig. \ref{fig:bc-bt}.
        Thus, the success and failure pathways for a BC-BT are given by $\mathfrak{S} = \left\{i \in A \middle| \lambda_i =\emptyset\right\}$ and $\mathfrak{F} = A \cup \left\{j \in C \middle | A_j = \emptyset\right\}$.

        % \begin{align}
        %     \mathfrak{S} &= \left\{i \in A \middle| \lambda_i =\emptyset\right\} \\
        %     \mathfrak{F} &= A \cup \left\{j \in C \middle | A_j = \emptyset\right\}
        % \end{align}

        According to (\ref{eq:operating}), we then have that $\Omega_i \cap S_i = \emptyset$ for all $i \in A$ such that $\lambda_i \neq \emptyset$.
        Additionally, since $\Omega_i \subseteq \bigcap_{j \in C_i} S_j$ and $\bigcap_{j \in C_i} S_j = B_i$, we have that $\Omega_i \cap F_i = \emptyset$ because $B_i \subseteq R_i \cup S_i$ (despite $i \in \mathfrak{F}$).
        According to (\ref{eq:operating}), we have that $\Omega_i \cap S_i = \emptyset$ for all $i \in C$.
        We also know that $R_i = \emptyset$ for all $i \in C$ by Def. \ref{def:libraries}, thus $\Omega_i \cap R_i = \emptyset$ for all $i \in C$.
        Lastly, Assumption \ref{as:bcbt} implies $\Omega_i \cap F_i = \emptyset$ for all $i \in C$.
        Thus, $\Omega_i = \emptyset$ for all $i \in C$.
    \end{proof}

    Since we have just shown that operating regions of conditions are empty, we can now compute the operating regions of actions explicitly (ignoring the conditions).
    To do so, we first need to compute the influence regions $I_i$ of such actions, following (\ref{eq:influence}).
    The set of nodes that determine an action's influence region can be partitioned by type: postconditions of the actions $C_i^\text{Post}$ (when they return success, the action is done), conditions that come before such postconditions $C_i^\text{ACC}$ (known as action-condition constraints (ACCs) in \cite{ogren2020convergence}), and preconditions of the action itself $C_i$.
    In the following definition,  we compute $C_i^\text{Post}$ and $C_i^\text{ACC}$.

    \begin{definition}[Related conditions]
        % The indicies of action sub-BTs ($\bt_{p(j)} : j \in A_i^\text{Pre}$) that need to fail for an action sub-BT ($\bt_{p(i)}$) to execute are
        % \begin{equation}
        %     A_i^\text{Pre} = \bigcup_{(a,b,c) \in \lambda_i} \left\{ j \in A \mid j <_{A_b} a \right\}.
        % \end{equation}
        % The indicies of action sub-BTs ($\bt_{p(j)} : j \in A_i^\text{Pro}$) that will execute if an action sub-BT ($\bt_{p(i)}$) fails are
        % \begin{equation}
        %     A_i^\text{Pro} = \bigcup_{(a,b,c) \in \lambda_i} \left\{ j \in A \mid a <_{A_b} j \right\}.
        % \end{equation}
        The indices of conditions ($j \in C_i^\text{Post}$) that need to fail for an action $\bt_i : i \in A$ to execute are
        \begin{equation}
            C_i^\text{Post} = \bigcup_{(a,b,c) \in \lambda_i} \left\{j \in C \mid j = b\right\}.
        \end{equation}
        The indices of conditions ($j \in C_i^\text{ACC}$) other than preconditions that need to succeed for an action ($\bt_i : i \in A$) to execute are
        \begin{equation}
            C_i^\text{ACC} = \bigcup_{(a,b,c) \in \lambda_i} \left\{ j \in C \mid j <_{C_c} b \right\}.
        \end{equation}
        % \begin{align}
        %     A_i^\text{Pre} &= \bigcup_{(a,b,c) \in \lambda_i} \left\{ j \in A \mid j <_{A_b} a \right\} \\
        %     A_i^\text{Pro} &= \bigcup_{(a,b,c) \in \lambda_i} \left\{ j \in A \mid a <_{A_b} j \right\} \\
        %     C_i^\text{Post} &= \bigcup_{(a,b,c) \in \lambda_i} \left\{j \in C \mid j = b\right\} \\
        %     C_i^\text{ACC} &= \bigcup_{(a,b,c) \in \lambda_i} \left\{ j \in C \mid j <_{C_c} b \right\}
        % \end{align}
    \end{definition}

    With the related actions, $C_i^\text{Post}$ and $C_i^\text{ACC}$, we are now ready to compute the influence regions, following (\ref{eq:influence}).

    \begin{lemma}
        % The influence region of a sub-BT in a BC-BT is
        % \begin{equation}
        %     I_i = \begin{cases}
        %         I_i' \bigcap_{j \in C_i} S_{p(i)} &\text{if} \quad i \in A \\
        %         I_k'\bigcap_{j <_{C_k} i} S_{p(j)} \mid \exists !k \in A : i \in C_k &\text{if} \quad i \in C,
        %     \end{cases}
        % \end{equation}
        % where
        % \begin{equation}
        %     I_i' := 
        %     \bigcap_{j \in A_i^\text{Pre}} F_{p(j)}
        %     \bigcap_{j \in C_i^\text{Post}} F_j
        %     \bigcap_{j \in C_i^\text{ACC}} S_{p(j)}
        % \end{equation}
        An action's ($i \in A$) influence region is
        % \begin{equation}
        %     I_i = I_i' \bigcap_{j \in C_i} S_{p(j)},
        % \end{equation}
        % where 
        \begin{equation}\label{eq:action_influence}
            % I_i' := \bigcap_{(a,b,c) \in \lambda_i} \left(
            %     \overbrace{F_b}^\text{Post.Cond.}
            %     \overbrace{\bigcap_{j <_{A_b} a} F_j^A}^\text{Prev.Acts.}
            %     \overbrace{\bigcap_{j <_{C_c} b} S_j^C}^\text{ACCs}
            % \right)
            I_i = 
            \bigcap_{j \in C_i^\text{ACC}} S_j
            \bigcap_{j \in C_i} S_j
            \bigcap_{j \in C_i^\text{Post}} F_j.
            % \bigcap_{j \in A_i^\text{Pre}} F_{p(j)}.
        \end{equation}
    \end{lemma}
    \begin{proof}

        The left uncles of actions $i \in A$ that have a fallback as their parent are given by postconditions $C_i^\text{Post}$.
        Thus, $\bt_j$ for all $j \in C_i^\text{Post}$ need to return failure for $\bt_i$ to execute.

        The left uncles of actions $i \in A$ that have a sequence as their parents are given by the preconditions $C_i$ and the parents of ACCs $C_i^\text{ACC}$.
        According to (\ref{eq:condition-metadata}),
        the success region of condition sub-BTs is equal to the success region of the conditions themselves.
        Thus, we have that $S_{p(j)} = S_j$ for all $j \in C$.
        Thus, $\bt_j$ for all $j \in C_i^\text{ACC} \cup C_i$ need to return success for $\bt_i$ to execute.

        The above completes the computation of (\ref{eq:influence}), thus we have (\ref{eq:action_influence}).
    \end{proof}

    We will now give an example of computing the influence regions based on the mobile manipulator BC-BT in Fig. \ref{fig:bc-bt}.

    \begin{example}
        For the BC-BT in Fig \ref{fig:bc-bt}, the influence regions of the actions are given by
        \begin{equation}
            \begin{array}{lllllll}
                &&\bigcap_{j \in C_i^\text{ACC}} S_j
                &&\bigcap_{j \in C_i} S_j
                &&\bigcap_{j \in C_i^\text{Post}} F_j \\
                \hline \\
                I_{9} & = & \mathbb{R}^n & \cap & S_2 & \cap & F_1\\
                I_{10} & = & S_1 & \cap & S_6 & \cap & F_3 \cap F_4 \cap F_5\\
                I_{11} & = & S_1 & \cap & S_5 & \cap & F_3 \cap F_4\\
                I_{12} & = & S_1 \cap S_4 & \cap & S_8 & \cap & F_3 \cap F_7\\
                I_{13} & = & S_1 & \cap & S_4 \cap S_7 & \cap & F_3\\
                I_{14} & = & \mathbb{R}^n & \cap & S_1 \cap S_3 & \cap & \mathbb{R}^n.
            \end{array}
        \end{equation}
        Based on Lemma \ref{lemma:backchained_operating}, we have that $\Omega_i = I_i \cap R_i$ for all $i \in \{9, 10, 11, 12, 13\}$ and $\Omega_i = I_i \cap (R_i \cup S_i)$ for $i = 14$.
        Note that $C_i^\text{Post} \cap (C_j^\text{ACC} \cup C_j)$ is non-empty when $i < j$ and empty when $i \geq j$.
    \end{example}
    
    In the above example, we noticed that there is a pattern in the condition indices.
    Namely, when the actions in a BC-BT are labeled with depth-first preorder traversal, we have that $C_i^\text{Post} \cap (C_j^\text{ACC} \cup C_j) \neq \emptyset \iff i < j$.
    This is because the job of each action is to bring the state to the success region of the action's postcondition, which serves as either a precondition or ACC for some further action.
    By design, this behavior is acyclical (see Corollary \ref{theorem:acyclical_convergence}) and fits Corollary \ref{cor:sprague2021continuous}, as we show in the following corollary.

    \begin{corollary}[Backchained BTs{\cite[Theorem 1, p.7]{ogren2020convergence}}]\label{cor:ogren2020convergence}
        A BC-BT satisfies Corollary \ref{cor:sprague2021continuous} with
        \begin{equation}\label{eq:backchained_convergence}
            \leq_{\Omega} \subseteq
            \left\{\left(i,j\right) \in A^2 \middle| C_i^\text{Post} \cap \left(C_j^\text{ACC} \cup C_j \right) \neq \emptyset \right\},
        \end{equation}
        if $B_i \subseteq \bigcap_{j \in C_i^\text{ACC}} S_j$ for all $i \in A$.
    \end{corollary}
    \begin{proof}
        For all actions $i \in A$, we have that $\Omega_i \subseteq B_i$ because $\bigcap_{j \in C} S_j = B_i$.
        Thus, if $x \in \Omega_i$ then $x \in B_i$.
        We know that $B_i$ is positively invariant under (\ref{eq:bt_execution}), (\ref{eq:bt_execution_discrete}). 
        Thus, since $B_i \subseteq \bigcap_{j \in C_i^\text{ACC}} S_j$, the state will never leave $\bigcap_{j \in C_i^\text{ACC}} S_j$ (and $\bigcap_{j \in C_i} S_i$) as long as it is in $\Omega_i$.

        However, $B_i$ can intersect $S_j$ for any $j \in C_i^\text{Post}$.
        Thus, from a given operating region $\Omega_i$, the state can only transition to another operating region $\Omega_j$ such that $C_i^\text{Post} \cap (C_j^\text{ACC} \cup C_j) \neq \emptyset$.
        Thus, we must have (\ref{eq:backchained_convergence}) in the context of Corollary \ref{cor:sprague2021continuous}.
    \end{proof}

 \subsection{Including data-driven controllers without harming convergence guarantees}
 \label{sec:NN_MB}
In this section, we will show how Theorem~\ref{theorem:general_convergence} captures the main result in \cite{sprague2022adding}.
The problem addressed in \cite{sprague2022adding} is how we can add a controller with unknown performance (e.g. \cite{sprague2020dynamic, sprague2020bath}) to an existing BT with convergence guarantees, without destroying those guarantees.

\begin{figure}[t]
    \resizebox{0.65\linewidth}{!}{%
    \begin{subfigure}[c]{0.48\linewidth}
        {\small
        \begin{forest}
            for tree={
                minimum height=2em,
                minimum width=2em,
                inner sep=0.1pt
            }
            [[$?$, draw
            [$\begin{array}{c}
                \text{Task} \\ \text{done}
            \end{array}$, draw, ellipse]
            [$\begin{array}{c}
                \text{MB} \\
                \text{controller}
            \end{array}$, draw]
            ]]
        \end{forest}
        }
    \end{subfigure}
    \begin{subfigure}[c]{0.48\linewidth}
        {\small
        \begin{forest}
            for tree={
                minimum height=2em,
                minimum width=2em,
                inner sep=0.1pt
            }
            [,
                [$?$, draw, 
                    [$\begin{array}{c}\text{Task}\\\text{done}\end{array}$, draw, ellipse]
                    [$\rightarrow$, draw
                        [$\begin{array}{c}\text{Time}\\\text{OK}\end{array}$, draw, ellipse]
                        [$\begin{array}{c}\text{Risk}\\\text{OK}\end{array}$, draw, ellipse]
                        [$\begin{array}{c}\text{DD}\\\text{Controller}\end{array}$, draw]
                    ]
                    [$\rightarrow$, draw
                        [$\begin{array}{c}\text{Time}\\\text{OK}\end{array}$, draw, ellipse]
                        [$\begin{array}{c}\text{RR}\\\text{Controller}\end{array}$, draw]
                    ]
                    [$\begin{array}{c}\text{MB}\\\text{controller}\end{array}$, draw]
                ]
            ]
        \end{forest}
        }
    \end{subfigure}
    }%
    \caption{To include a data-driven controller in a BT where an already existing model-based controller provides performance guarantees you can replace the left subtree with the right one.}
    \label{fig:NN_MB}
\end{figure}
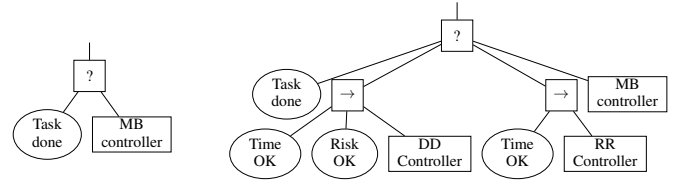

% \begin{figure}
%     \centering
%     \includegraphics[width=\columnwidth]{img/NN_MB.pdf}
%     \caption{To include a data-driven controller in a BT where an already existing model-based controller provides performance guarantees you can replace the left subtree with the right one.}
%     \label{fig:NN_MB}
% \end{figure}

The key idea is illustrated in Fig.~\ref{fig:NN_MB}. Imagine that the data-driven controller is designed to perform the same task as an existing model-based controller, in an existing BT, satisfying Theorem~\ref{theorem:general_convergence}.
Now, we think that the data-driven controller is more efficient in most cases, but perhaps %it is
not all. Thus, we replace the subtree on the left of Fig.\ref{fig:NN_MB} with the larger subtree on the right of Fig.\ref{fig:NN_MB}. There are two new conditions, \emph{Time OK} and \emph{Risk OK}. \emph{Time OK} returns success if the amount of time spent executing the new controller is still smaller than some designated upper bound (how long we are willing to let the new controller execute before reverting to the old one we know will do the job). \emph{Risk OK} is a bit more complex and connected to the additional controller \emph{Reduce risk}. These two are making sure that the \emph{Data-Driven Controller} does not mess anything up by creating a transition to an operating region that was not possible from the operating region of the model-based controller in the old design. 

Imagine the original subtree had operating region $\Omega_{\text{old}}$ and the new subtree has operating region $\Omega_{\text{new}}$. If we can show that: $\Omega_{\text{new}}=\Omega_{\text{old}}$, any transitions out of $\Omega_{\text{new}}$ for the new BT was also possible for the old BT, and a transition out of $\Omega_{\text{new}}$ will happen in finite-time; then we are done.

\begin{assumption}
    \label{ass:NN_MB}
    Given a BT that satisfies the conditions stated in Theorem~\ref{theorem:general_convergence},
     let $\Omega_i$ be the operating region of some subtree of the BT of the form of the left of Fig.\ref{fig:NN_MB}, and $\Omega_{i,\text{next}} := \bigcup \{\Omega_j \mid (i,j) \in E\}$, i.e., the union of  the nodes in the prepares graph $\Gamma$ (Def. \ref{def:prepares_graph}) that have an incoming edge from $i$.
     
      A new BT is created by replacing the left subtree by the one on the right of Fig.\ref{fig:NN_MB}, with three controllers, the original 
    \emph{Model-Based controller} (MB), the
    \emph{Data-Driven controller} (DD), and the 
    \emph{Reduce Risk controller} (RR),
    with conditions \emph{Time OK}, \emph{Risk OK}, and \emph{Task Done} abbreviated as \emph{TOK}, \emph{ROK}, and \emph{TD}, respectively. 
\end{assumption}

\begin{assumption}
    \label{ass:NN_MB_2}
If Assumption \ref{ass:NN_MB} holds, 
    assume that the resulting prepares graph $\Gamma_{\text{new}}$ is identical to the previous one $\Gamma_{\text{old}}$, except that the subgraph on the left of Fig.~\ref{fig:NN_MB_convergence} is replaced by the subgraph on the right of Fig.~\ref{fig:NN_MB_convergence}.
    
    Furthermore assume that the state is guaranteed not to be in either $v_a(DD)$ or $v_b(RR)$ after some finite time and that $R_i,S_i,F_i$ are the same before and after changing the subtrees.
\end{assumption}

\begin{lemma}
\label{lemma:NN_MB}
    If Assumptions~\ref{ass:NN_MB} and \ref{ass:NN_MB_2} hold, the new BT is also convergent.
\end{lemma}

\begin{proof}
    We start by noting that $\Omega_i$ is the same before and after the change, since $I_i$ only depends on the Left-uncles of the subtree, see (\ref{eq:influence}), and $R_i, S_i, F_i$ are the same by Assumption \ref{ass:NN_MB_2}.
    %, see (\ref{eq:operating}). 
  
Thus, since the operating region of a node is the union of the operating regions of its children \cite[Lemma 3]{sprague2021continuous} we have
 that 
    $\Omega_{new} = \Omega_{DD} \cup \Omega_{RR} \cup \Omega_{MB}$.

By Assumption~\ref{ass:NN_MB_2} the prepares graph $\Gamma_{\text{new}}$ is identical, with the same transitions, except for the new subgraph.
The new subgraph has a loop, but again by Assumption~\ref{ass:NN_MB_2}, this loop will be left in finite time.

Thus, the argument that was made to prove that the original BT was convergent can be applied to the new BT.
\end{proof}

 \begin{figure}
        \centering
        \begin{tikzpicture}[main/.style = {draw}]
            \node[main, very thick] (DD) {$v_a(DD)$};
            \node[main, very thick] (RR) [right=1cm of DD] {$v_b(RR)$};
            \node[main, very thick] (MB) [right=1cm of RR] {$v_b(MB)$};
            \node[main, very thick] (NEXT) [below=1cm of RR] {$v_b(i, \text{next})$};
            \node[main, draw=none] (DDup) [above=1cm of DD] {};
            \node[main, draw=none] (RRup) [above=1cm of RR] {};
            \node[main, draw=none] (MBup) [above=1cm of MB] {};
            \node[main, draw=none] (NEXTdown) [below=1cm of NEXT] {};
            \draw[->] (DD) -- (RR);
            \draw[->] (RR) -- (MB);
            \draw[->] (DD) -- (NEXT);
            \draw[->] (RR) -- (NEXT);
            \draw[->] (MB) -- (NEXT);
            \draw[->] (RR) to [out=200, in=-20] (DD);
            \draw[->] (DD) to [out=30, in=150] (MB);
            \draw[->, dotted] (DDup) to (DD);
            \draw[->, dotted] (RRup) to (RR);
            \draw[->, dotted] (MBup) to (MB);
            \draw[->, dotted] (NEXT) to (NEXTdown);
            % "Before picture to the left
            \node[main, very thick] (MB2) [left=1cm of DD] {$v_b(MB)$};
            \node[main, very thick] (NEXT2) [below=1cm of MB2] {$v_b(i, \text{next})$};
            \node[main, draw=none] (MBup2) [above=1cm of MB2] {};
            \node[main, draw=none] (NEXTdown2) [below=1cm of NEXT2] {};
            \draw[->, dotted] (MBup2) to (MB2);
            \draw[->] (MB2) -- (NEXT2);
            \draw[->, dotted] (NEXT2) to (NEXTdown2);
        \end{tikzpicture}
        \caption{Parts of the two prepares graphs $\Gamma$ 
        corresponding to the two BTs in Fig. \ref{fig:NN_MB}. Note that when Time OK returns Failure a transition away from  $v_a(DD)$ and $v_b(RR)$ has to happen.}
        \label{fig:NN_MB_convergence}
 \end{figure}
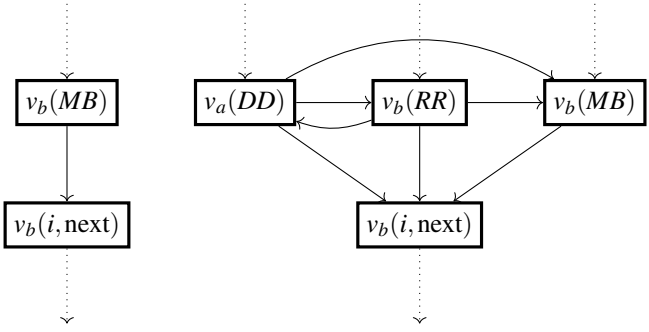

\begin{lemma}
    If Assumption~\ref{ass:NN_MB} holds, then Assumption~\ref{ass:NN_MB_2} also holds if the following is true.

%Let    $\Omega'=\Omega_{DD} \cup \Omega_{RR} \cup \Omega_{MB} \cup \Omega_{iNext}$

The success region of MB is inside TD, 
the running region of DD is the state space, 
the success region of RR is inside ROK,  and 
the model-based controller does not fail inside $F_{TD}$
i.e.,
\begin{align}
&S_{MB} \subset S_{TD} \\ 
&R_{DD} = \mathbb{R}^n \\ 
&S_{RR} \subset S_{ROK} \\ 
&F_{TD} \cap F_{MB} = \emptyset.
\end{align}

The condition Time OK stops returning success after some time $t'$. Note that this condition can be created as a function of state by e.g., letting some component of the state $x \in \mathbb{R}^n$ have the dynamics $\dot x_k = 1,~x_k(0)=0$ and checking $x_k(t) < t'$.

The ROK condition is designed such that $v_b(RR)$ surrounds most of $v_a(DD)$, and the only neighbors, in the sense of Definition \ref{def:neighboring_sets}, to $v_a(DD)$ are 
$v_b(RR), v_b(MB)$ and $v_b(i,\text{next})$.

Finally, the MB and RR controllers are well-behaved, in the sense that 
$v_b(RR)=\Omega_{RR}$ and
$v_b(MB)=\Omega_{MB}$.

\end{lemma}

\begin{proof}
We start by showing that the metadata regions are the same for the new and old subtrees.

The success region of the old BT is given by
\begin{equation}
    S_{\text{old}} = S_{TD} \cup \left(F_{TD} \cap S_{MB} \right).
\end{equation}
The assumption that $S_{MB} \subset S_{TD}$ and metadata regions are pairwise disjoint implies $S_{\text{MB}} \cap F_{\text{TD}} = \emptyset$, and thus $S_{\text{old}} = S_{\text{TD}}$.
The success region of the new BT is given by
    \begin{equation}
    \begin{aligned}
    S_{\text{new}} = \, & S_{TD} \\
    & \cup \left( F_{TD} \cap S_{DD} \cap S_{ROK} \cap S_{TOK} \right) \\
    & \cup \left( F_{TD} \cap S_{MB} \cap \left( F_{TOK} \cup \left( F_{RR} \cap S_{TOK} \right) \right) \right. \\
    & \quad \left. \cap \left( F_{TOK} \cup \left( F_{ROK} \cap S_{TOK} \right) \cup \left( F_{DD} \cap S_{ROK} \cap S_{TOK} \right) \right) \right) \\
    & \cup \left( F_{TD} \cap S_{RR} \cap S_{TOK} \cap \left( F_{TOK} \cup \left( F_{ROK} \cap S_{TOK} \right) \right. \right. \\
    & \quad \left. \left. \cup \left( F_{DD} \cap S_{ROK} \cap S_{TOK} \right) \right) \right).
    \end{aligned}
\end{equation}
The assumption $R_{DD} = \mathbb{R}^n$ and the definition that metadata regions are pairwise disjoint implies $S_{DD} = \emptyset$ and therefore the second unioned term can be eliminated.
The assumption $S_{MB} \subset S_{TD}$ implies that the third unioned term can be eliminated.
Now consider the distributed terms of the intersection $F_{TD} \cap S_{RR} \cap S_{TOK}$ in the fourth unioned term.
The definition that metadata regions are pairwise disjoint, so that $S_{TOK} \cap F_{TOK} = \emptyset$, implies that the first distributed term is empty.
The assumption $S_{RR} \subset S_{ROK}$ implies that $S_{ROK}$ can be eliminated in the third distributed term.
This assumption in addition to the definition that metadata regions are pairwise disjoint implies $S_{RR} \cap F_{ROK} = \emptyset$, so that the second distributed term is empty.
Lastly, $R_{DD} = \mathbb{R}^n$ implies $F_{DD} = \emptyset$, so we end up with $S_{\text{new}} =  S_{TD}$.
Thus $S_{\text{old}} = S_{\text{new}}$.

The failure region of the old BT is given by
\begin{equation}
    F_{\text{old}} = F_{TD} \cap F_{MB}.
\end{equation}
The assumption $F_{TD} \cap F_{MB}$ immediately implies $F_{\text{old}} = \emptyset$.
The failure region of the BT is given by
\begin{equation}
    \begin{aligned}
        F_{\text{new}} =& F_{MB} \cap F_{TD} \cap \left(F_{TOK} \cup \left(F_{RR} \cap S_{TOK}\right)\right) \\
        &\cap \left(F_{TOK} \cup \left(F_{ROK} \cap S_{TOK}\right) \cup \left(F_{DD} \cap S_{ROK} \cap S_{TOK}\right)\right).
    \end{aligned}
\end{equation}
The assumption $F_{TD} \cap F_{MB}$ immediately implies $F_{\text{new}} = \emptyset$, as before.
Thus $F_{\text{old}} = F_{\text{new}}$.

Lastly, the definition that metadata regions are pairwise disjoint immediately implies $R_{\text{new}} = R_{\text{old}} = \mathbb{R}^n \setminus (S_{\text{old}} \cup F_{\text{old}})$.
Thus, the metadata regions of the old and new BT are equal.

Now we show that the state is guaranteed not to be in either $v_a(DD)$ or $v_b(RR)$ after time $t'$. This is clear as Time OK returns failure of $t>t'$, and is a precondition of those corresponding actions in the BT of Fig. \ref{fig:NN_MB}.

Finally, we show that the prepares graph $\Gamma_{\text{new}}$ has the transitions shown in Fig. \ref{fig:NN_MB_convergence}.

Transitions can only happen between neighboring operating regions, see Definition \ref{def:neighboring_sets},
so even though we have no control of the vector field of the \emph{Data Driven controller} we know that the only neighbors to $v_a(DD)$ are 
$v_b(RR), v_b(MB), v_b(i,\text{next})$, so the transitions out of $v_a(DD)$ are as in the figure.

Furthermore, we know that 
$S_{RR} \subset S_{ROK}$ so the transition out of $v_b(RR)$ is to either $v_a(DD)$ if Time OK or to $v_b(MB)$ if Time is not ok.
Similarly, we know that 
$S_{MB} \subset S_{TD}$ so the transition out of $v_b(MB)$ is to $v_b(i, \text{next})$.
\end{proof}

\begin{remark}
    Let $i$ be the root index of the old and new BT, equivalently. 
    The new BT has the following regions of influence:
    \begin{align}
        I_{\text{DD}} =& I_i \cap F_{\text{TD}} \cap S_{\text{TOK}} \cap S_{\text{ROK}} \\
        I_{\text{RR}} =& I_i \cap F_{\text{TD}} \cap S_\text{TOK} \cap \nonumber \\ 
        &\left[F_\text{TOK} \cup \left(F_\text{ROK} \cap S_\text{TOK} \right) \cup  \left(F_\text{DD} \cap S_\text{TOK} \cap S_\text{ROK}\right)\right] \nonumber \\
        =& I_i \cap F_\text{TD} \cap S_\text{TOK} \cap \left[F_\text{TOK} \cup \left(F_\text{ROK} \cup F_\text{DD}\right)\right] \nonumber \\
        =& I_i \cap F_\text{TD} \cap S_\text{TOK} \cap \left(F_\text{ROK} \cup F_\text{DD}\right) \\
        I_{\text{MB}} =& I_i \cap F_\text{TD} \cap \left[F_\text{TOK} \cup \left(F_\text{RR} \cap S_\text{TOK}\right)\right] \cap \nonumber \\
        &\left[F_\text{TOK} \cup \left(F_\text{ROK} \cap S_\text{TOK} \right) \cup  \left(F_\text{DD} \cap S_\text{TOK} \cap S_\text{ROK}\right)\right] \nonumber \\
        =& I_i \cap F_\text{TD} \cap \left(F_\text{TOK} \cup F_\text{RR}\right) \cap \left[F_\text{TOK} \cup \left(F_\text{ROK} \cup F_\text{DD}\right)\right] 
    \end{align}
    where $I_i$ is given, and we have used $S_\text{TOK} \cap F_\text{TOK} = \emptyset$ and $S_\text{ROK} \cup F_\text{ROK} = \mathbb{R}^n$, i.e. metadata regions are pairwise disjoint and cover $\mathbb{R}^n$, and conditions have empty running regions.
    
    Given that $S_i = S_{TD}$ and $F_i=\emptyset$,
    this makes the corresponding operating regions
    \begin{align}
        \Omega_{\text{DD}} &=I_{\text{DD}} \cap R_{\text{DD}}\\
        \Omega_{\text{RR}} &=I_{\text{RR}} \cap R_{\text{RR}}\\
        \Omega_{\text{MB}} &=I_{\text{MB}} \cap R_{\text{MB}}.
    \end{align}

\end{remark}

If we want to introduce hysteresis in the switching between the \emph{Data Driven controller} and \emph{Reduce Risk}, to avoid chattering and provide a better starting state for the \emph{Data Driven controller}, we can make use of the following result
\begin{lemma}
    If we replace the \emph{Risk OK} condition with the conjunction (\emph{Risk OK} and \emph{Risk OK for T seconds}). The result above still holds.
\end{lemma}

\begin{proof}
    Without loss of generality, we assume that the state $x\in \mathbb{R}^n$ includes a component that enables us to evaluate the condition \emph{Risk OK for T seconds}, see  Remark~\ref{rem:Tsec} below.
    This change will only affect
    $\Omega_{DD}$ and $\Omega_{RR}$, making the former smaller and the latter bigger, but leaving all other operating regions, as well as $\Omega_{DD} \cup \Omega_{RR}$ unchanged.

Given this observation, it is clear that the requirements in Lemma~\ref{lemma:NN_MB} still hold.
\end{proof}

\begin{remark}
\label{rem:Tsec}
    In order to create an approximation of the condition: 
    \emph{Risk OK for T seconds} that is accessible by a simple state check, we can extend the state with an extra component as follows.
    Then $x_{n+1}\geq 1$ approximately corresponds to \emph{Risk OK for T seconds}.  
$$
  x_{n+1}(t+1)= 
\begin{cases}
    x_{n+1}(t) + \Delta t/T,& \text{if } x \in S_{\text{ROK}} \land x_{n+1}<1\\
    0,          & \text{else}
\end{cases}
$$
For the continuous-time case, we get an approximation through
$$
  \dot x_{n+1}= 
\begin{cases}
    (1-x_{n+1})e^T& \text{if } x_{n+1} \in S_{\text{ROK}} \\
    -Kx_{n+1}          & \text{else}
\end{cases}
$$
for some large constant $K$.
\end{remark}

    \section{Conclusions} 
    \label{sec:conclusions}
    In this paper, we presented a general convergence theorem (Theorem \ref{theorem:general_convergence}) that goes beyond earlier results in that it allows cycles in the execution of sub-BTs and the execution of sub-BTs outside their DOA. 
    We also showed how this theorem generalizes earlier proofs in \cite{colledanchise2016behavior,ogren2020convergence,sprague2021continuous}. In detail, 
        \cite[Lemma 2, p.7]{colledanchise2016behavior} is covered by Corollary \ref{cor:sequence},
        \cite[Lemma 3, p.8]{colledanchise2016behavior} is covered by Corollary \ref{cor:implicit_sequence},
        \cite[Theorem 4, p.6]{sprague2021continuous} is covered by Corollary \ref{cor:sprague2021continuous}, and
        \cite[Theorem 1, p.7]{ogren2020convergence} is covered by Corollary \ref{cor:ogren2020convergence}.
        Finally, the main theorem also captures the result of combining data-driven and model-based controllers without sacrificing performance guarantees, from \cite{sprague2022adding}.
    Thus, this paper provides a unified and extended analysis of behavior tree convergence.
    
    \section*{Acknowledgments}
    This  work  was  supported  by SSF  through  the  Swedish  Maritime Robotics Centre (SMaRC) (IRC15-0046), and FOI project 7135.
    
    \bibliographystyle{IEEEtran}
    \bibliography{root.bib}

\end{document}